\documentclass[10pt]{article}
\usepackage{latexsym}
\usepackage{a4wide,url}
\usepackage{amssymb}
\usepackage{amsmath}
\usepackage{amsthm}
\usepackage{amsfonts}

\theoremstyle{plain}
\newtheorem{theorem}{Theorem}[section]
\newtheorem{corollary}[theorem]{Corollary}
\newtheorem{proposition}[theorem]{Proposition}
\newtheorem{lemma}[theorem]{Lemma}
\newtheorem{claim}[theorem]{{\sc Claim}}

\newcommand{\auf}{\langle}
\newcommand{\zu}{\rangle}
\newcommand{\C}{\textit{Column}}
\newcommand{\cl}{\textit{cl}}
\newcommand{\mprod}{\!\times\!}
\newcommand{\eprod}{\!\times^{\textit{e}}\!}
\newcommand{\dprod}{\!\times^{\textit{d}}\!}

\newcommand{\subf}{\textit{sub}\,\phi}
\newcommand{\hr}{\textit{hr}}
\newcommand{\Ix}{I}
\newcommand{\preorder}{weak order\ }
\newcommand{\preorders}{weak orders\ }

\newcommand{\logic}{{\sc FOLTL}${}^{\ne}$}
\newcommand{\model}{{\sc FOLTL}-model}

\newcommand{\Log}{{\sf Log}}
\newcommand{\fusion}{L_0\oplus L_1}
\newcommand{\commut}{[L_0,L_1]}

\newcommand{\Kfour}{\mathbf{K4}}
\newcommand{\Kfourt}{\mathbf{K4.3}}

\newcommand{\Dis}{\mathbf{DisK4.3}}

\newcommand{\Sfive}{\mathbf{S5}}
\newcommand{\Diff}{\mathbf{Diff}}

\newcommand{\lindiffcomm}{[\Kfourt,\Diff]}

\newcommand{\Fr}{{\sf Fr}\,}
\newcommand{\CC}{\mathcal{C}}
\newcommand{\cdiff}{\mathcal{C}_{\textit{diff}}}
\newcommand{\clin}{\mathcal{C}_{\textit{lin}}}
\newcommand{\cdifff}{\mathcal{C}_{\textit{diff}}^\textit{fin}}
\newcommand{\clinf}{\mathcal{C}_{\textit{lin}}^\textit{fin}}

\newcommand{\F}{\mathfrak{F}}
\newcommand{\G}{\mathfrak{G}}
\newcommand{\Hh}{\mathfrak{H}}

\newcommand{\M}{\mathfrak{M}}
\newcommand{\N}{\mathfrak{N}}
\newcommand{\Mrec}{\mathfrak{M}^{\textit{rec}}}
\newcommand{\Mfin}{\mathfrak{M}^{\,\textit{fin}}}
\newcommand{\Minf}{\mathfrak{M}^\infty}

\newcommand{\Ninf}{\mathfrak{N}^\infty}

\newcommand{\existsx}{\exists x\,}
\newcommand{\forallx}{\forall x\,}
\newcommand{\foralld}{\forall^{\ne}\! x\,}
\newcommand{\existsd}{\exists^{\ne}\! x\,}
\newcommand{\existso}{\exists^{=1}x\,}
\newcommand{\Dt}{\Diamond_{\!F}}
\newcommand{\Bt}{\Box_{\!F}}

\newcommand{\D}{\Diamond}
\newcommand{\B}{\Box}
\newcommand{\Bh}{\B_{0}}
\newcommand{\Dh}{\D_{0}}
\newcommand{\Bv}{\B_{1}}
\newcommand{\Dv}{\D_{1}}
\newcommand{\DhM}{\blacklozenge_0}
\newcommand{\BhM}{\blacksquare_0}

\newcommand{\E}{\D^{\bf\tiny =1}}
\newcommand{\Ev}{\Dv^{\bf\tiny =1}}

\newcommand{\nexttime}{X}

\newcommand{\recx}{{\sf R}}
\newcommand{\recp}{{\sf S}^\ast}
\newcommand{\cou}{{\sf C}_i}
\newcommand{\pred}{{\sf P}}
\newcommand{\predq}{{\sf Q}}
\newcommand{\predp}{{\sf P}'}
\newcommand{\fodompred}{{\sf D}}
\newcommand{\pint}{{\sf Interval}_{{\sf P}}}
\newcommand{\iotaint}{{\sf Interval}_{{\sf I}_\iota}}
\newcommand{\sint}{{\sf Interval}_{\state}}
\newcommand{\dint}{{\sf Interval}_{\diag}}

\newcommand{\cint}{{\sf Interval}_{{\sf C}_i^-}}
\newcommand{\cinti}{{\sf Interval}_{{\sf C}_i}}
\newcommand{\plusi}{{\sf C}^+_i}
\newcommand{\minusi}{{\sf C}^-_i}
\newcommand{\minusip}{{\sf C}^{-'}_i}
\newcommand{\start}{{\sf start}}
\newcommand{\pend}{{\sf end}}
\newcommand{\psivar}{{\sf P}_{\!\psi}}

\newcommand{\insbw}{{\sf I}_{\iota}}

\newcommand{\Rh}{R_0}
\newcommand{\Rv}{R_1}
\newcommand{\RhM}{R^{\M}\!}
\newcommand{\Rhp}{\tilde{\Rh}}
\newcommand{\Rvp}{\tilde{\Rv}}

\newcommand{\diag}{{\sf N}}
\newcommand{\state}{{\sf S}}
\newcommand{\diagbw}{{\sf grid}^{\textit{bw}}}
\newcommand{\diagbwd}{{\sf grid}^\ast}
\newcommand{\diagfw}{{\sf grid}}
\newcommand{\diagfwu}{{\sf grid\_unique}}
\newcommand{\diagfwfin}{{\sf grid}_{\textit{fin}}}
\newcommand{\diagfwfinu}{{\sf grid\_unique}_{\textit{fin}}}

\newcommand{\allc}{{\sf AllC}_i}
\newcommand{\allcl}{{\sf TillStartAllC}_i}
\newcommand{\allcd}{{\sf AllC}_i^\bullet}
\newcommand{\couf}{{\sf counter}^{\,\textit{bw}}}
\newcommand{\couffw}{{\sf counter}}
\newcommand{\coufd}{{\sf counter}^{\,\textit{bw}\bullet}}

\newcommand{\ffix}{{\sf Fix}_i}
\newcommand{\ffixbwd}{{\sf Fix}_i^{\textit{bw}\bullet}}
\newcommand{\ffixbw}{{\sf Fix}_i^{\textit{bw}}}
\newcommand{\ffixbwl}{{\sf Fix}_i^{\textit{lossy}}}
\newcommand{\ffixjbwl}{{\sf Fix}_j^{\textit{lossy}}}
\newcommand{\ffixjbw}{{\sf Fix}_j^{\textit{bw}}}
\newcommand{\ffixj}{{\sf Fix}_j}
\newcommand{\fdec}{{\sf Dec}_i}
\newcommand{\fdecbwd}{{\sf Dec}_i^{\textit{bw}\bullet}}
\newcommand{\fdecbw}{{\sf Dec}_i^{\textit{bw}}}
\newcommand{\fdecbwl}{{\sf Dec}_i^{\textit{lossy}}}

\newcommand{\finc}{{\sf Inc}_i}
\newcommand{\fincbwd}{{\sf Inc}_i^{\textit{bw}\bullet}}
\newcommand{\fincbw}{{\sf Inc}_i^{\textit{bw}}}
\newcommand{\fincbwl}{{\sf Inc}_i^{\textit{lossy}}}

\newcommand{\exei}{{\sf Do}_\iota}
\newcommand{\exeibw}{{\sf Do}_\iota^{\textit{bw}}}
\newcommand{\exeibwd}{{\sf Do}_\iota^{\textit{bw}\bullet}}
\newcommand{\exeil}{{\sf Do}_\iota^{\textit{lossy}}}
\newcommand{\fmm}{\varphi_M}
\newcommand{\fmmbwdec}{\varphi_M^{\textit{bw}}}
\newcommand{\fmmbwdecda}{\varphi_M^{\textit{bw}\bullet}}
\newcommand{\fmmbwlossy}{\varphi_M^{\textit{lossy}}}
\newcommand{\fmmbwdecd}{\varphi_M^\ast}

\newcommand{\tick}{{\sf Tick}}

\newcommand{\frecbw}{{\sf rec}^{\textit{bw}}}
\newcommand{\frece}{{\sf rec}}

\newcommand{\expform}{\chi_\phi}

\newcommand{\finci}{c_i^{++}}
\newcommand{\fdeci}{c_ i^{--}}
\newcommand{\ftest}{c_i^{??}}
\newcommand{\stepi}{\mathop{\to}^{\iota}}
\newcommand{\stepin}{\mathop{\to}^{\iota_n}}
\newcommand{\step}{\mathop{\to}}
\newcommand{\lstepi}{\mathop{\to}_{\textit{\scriptsize lossy}}^{\iota}}
\newcommand{\lstep}{\mathop{\to}_{\textit{\scriptsize lossy}}}

\begin{document}


\title{Undecidable propositional bimodal logics and\\ one-variable first-order linear temporal logics with counting}

\author{C.~Hampson and A.~Kurucz\\
{\small Department of Informatics, King's College London}}

\maketitle


\begin{abstract}
First-order temporal logics are notorious for their bad
computational behaviour. It is known that even the two-variable
monadic fragment is highly undecidable over various linear timelines, and over branching time even one-variable fragments might be undecidable.
However, there have been several attempts on finding well-behaved fragments of 
first-order temporal logics and related temporal description logics, mostly either by restricting
the available quantifier patterns, or considering sub-Boolean languages.
Here we analyse seemingly `mild' extensions of decidable one-variable fragments with counting capabilities,
interpreted in models with constant, decreasing, and expanding first-order domains.
We show that over most classes of linear 
orders these logics are (sometimes highly) undecidable, even without constant and function symbols, and
with the sole temporal operator `eventually'.

We establish connections with bimodal logics over 2D product structures having linear and `difference' 
(inequality) component
relations,
and prove our results in this bimodal setting.  We show a
general result saying that satisfiability over many classes of bimodal models with commuting 
`unbounded' linear and difference relations is undecidable.
As a by-product, we also obtain 
new examples of finitely axiomatisable but Kripke incomplete bimodal logics.
Our results generalise similar lower bounds on bimodal logics over products of two linear relations, and
our proof methods are quite different from the known proofs of these results.
Unlike previous proofs that first `diagonally encode' an infinite grid, 
and then use reductions of tiling or Turing machine problems, 
here we make direct use of the grid-like structure of product frames and obtain lower complexity bounds
 by reductions of
counter (Minsky) machine problems. 
Representing counter machine runs apparently requires less control over 
neighbouring grid-points than tilings or Turing machine runs, and so this technique is possibly more versatile,
even if one component of the underlying product structures is `close to' being the universal relation.
\end{abstract}








\section{Introduction}\label{intro}


\subsection{First-order linear temporal logic with counting.}\label{introfoltl}

Though first-order temporal logics are natural and expressive languages for querying and constraining
temporal databases \cite{Chomicki94,Chomicki&Niwinski95} and reasoning about knowledge that changes in
time \cite{Hodkinsonetal01lpar},
their practical use has been discouraged by their high computational complexity.
It is well-known that even the two-variable monadic fragment is undecidable over various linear timelines,
and its satisfiability problem is $\Sigma_1^1$-hard over the natural
numbers \cite{Sz86,Szh88,Merz92,Gabbayetal94,gkwz03}. Also, even the one-variable fragment of first-order branching time 
logic $CTL^\ast$ is undecidable \cite{HodkinsonWZ02}.
Still, similarly to classical first-order logic where the decision problems of its fragments were 
studied and classified in
great detail \cite{bgg97}, there have been a number of attempts on finding the border between decidable and undecidable fragments of 
first-order temporal logics and related temporal description logics, mostly either by restricting
the available quantifier patterns  
\cite{Chomicki&Niwinski95,Hodkinsonetal00,Hodkinsonetal01lpar,Bauer&Hodkinson&W&Z04,Deg&Fisher&Konev06,Hodkinson02,Hodkinson06,Mamouras09}, 
or considering
sub-Boolean languages \cite{lwz08,akrz14}.

In this paper we contribute to this `classificational' research line by considering seemingly `mild' extensions of decidable one-variable fragments.
We study the satisfiability problem of the one-variable `future' fragment of linear temporal logic with counting to two,
interpreted in models over various timelines, and having constant, decreasing, or expanding first-order domains. 
Our language \logic\ keeps all Boolean connectives, it has no restriction on formula-generation, and it
is strong enough to express uniqueness of a property of domain
elements ($\existso$), and the `elsewhere' quantifier ($\foralld$). However,
\logic-formulas use only a single variable (and so contain only monadic predicate symbols), \logic\ has no equality, no constant or function symbols, and its only temporal
operators are `eventually' and  `always in the future'. 
\logic\ is weaker than the two-variable monadic \emph{monodic} fragment with equality,
where temporal operators can be applied only to subformulas with at most one free variable.
(This fragment with the `next time' operator
is known to be $\Sigma_1^1$-hard over the natural numbers \cite{Wolter&Z02apal,Degtyarevetal02}.)
\logic\ is connected to bimodal product logics \cite{Gabbay&Shehtman98,gkwz03} (see also below),
and to the temporalisation of the expressive description logic
$\mathcal{CQ}$ with one global universal role \cite{Wolter&Z99fund}.
Here are some examples of \logic-formulas:
\begin{itemize}
\item
``\emph{An order can only be submitted once:}''\quad
$
\forallx\Bt\bigl({\sf Subm}(x)\to\Bt\neg{\sf Subm}(x)\bigr).
$
\item
The Barcan formula:\quad
$
\existsx\Dt\pred(x)\leftrightarrow\Dt\existsx\pred(x).
$
\item
``\emph{Every day has its unique dog:}''\quad
$
\Bt\existso{\sf Dog}(x)\land\Bt\forallx\bigl({\sf Dog}(x)\to\Bt\neg{\sf Dog}(x)\bigr).
$
\item
``\emph{It's only me who is always unlucky:}''\quad
$
\Bt\neg{\sf Lucky}(x)\land \foralld\Dt{\sf Lucky}(x).
$
\end{itemize}
Note that \logic\ can also be considered as a fragment of three-variable classical first-order logic with only binary predicate symbols, but it is not within the guarded fragment.

\paragraph{Our contribution}
While the addition of `elsewhere' quantifiers to the two-variable fragment of classical first-order logic
does not increase the {\sc NExpTime} complexity of its satisfiability problem
\cite{Graedeletal97,gor97,PacholskiST00},
we show that adding the same feature to the (decidable) one-variable
fragment of first-order temporal logic results in (sometimes highly) undecidable logics over most linear timelines,
not only in models with constant domains, but even those with decreasing and expanding first-order domains.
Our main results on the \logic-satisfiability problem are summarised in Fig.~\ref{f:results}.

\begin{figure}[h]
\begin{center}
\begin{tabular}{l||c|c|c|c|}
& $\langle\omega,<\rangle$ & all finite  & all & $\langle\mathbb Q,<\rangle$ \\
& & linear orders & linear orders &  \mbox{or} $\langle\mathbb R,<\rangle$\\
\hline\hline
&&&&\\
constant & $\Sigma_1^1$-complete &  undecidable  & undecidable & undecidable \\
\ \ domains &  & r.e. & co-r.e. &\\[3pt]
\ \ & {\scriptsize Cor.~\ref{co:undecddisc}} & {\scriptsize Cor.~\ref{co:undecddisc}} & {\scriptsize Cor.~\ref{co:mainfo}} & {\scriptsize Cor.~\ref{co:densefo}}\\[3pt]
\hline
&&&&\\
decreasing  & $\Sigma_1^1$-complete &  undecidable  & undecidable & undecidable\\
\ \ domains &  & r.e. & co-r.e. &\\[3pt]
\ \ & {\scriptsize Cor.~\ref{co:undecddisc}} & {\scriptsize Cor.~\ref{co:undecddisc}} & {\scriptsize Cor.~\ref{co:mainfod}} & {\scriptsize Cor.~\ref{co:densefo}}\\[3pt]
\hline
&&&&\\
expanding &  undecidable & Ackermann-hard& \framebox{decidable?} & \framebox{decidable?}\\
\ \ domains &  co-r.e. &  \qquad decidable & co-r.e. &  \\[3pt]
&  {\scriptsize Cors.~\ref{co:omegaefol}, \ref{co:omegaefou}} & {\scriptsize Cors.~\ref{co:decefinfol}, \ref{co:decefinfou}} & {\scriptsize Cor.~\ref{co:linefou}}& \\[3pt]
\hline
\end{tabular}
\end{center}
\caption{\logic-satisfiability over various timelines and first-order domains.}\label{f:results}
\end{figure}


\subsection{Bimodal logics and two-dimensional modal logics.}\label{intromod}

It is well-known that the first-order quantifier $\forall x$ can be considered as an `$\Sfive$-box': a
propositional modal necessity operator interpreted over relational structures 
$\auf W,R\zu$ where $R=W\mprod W$
 (\emph{universal frames\/}, in modal logic parlance).
Therefore, the two-variable fragment of classical first-order logic is related to propositional bimodal logic over
two-dimensional (2D) \emph{product frames} \cite{Marx&Venema97}.
Similarly, the `elsewhere' quantifier $\foralld$ can be regarded as a `$\Diff$-box': a
propositional modal necessity operator interpreted over \emph{difference} frames $\auf W,\ne\zu$ where $\ne$ is the
inequality relation on $W$. Looking at \logic\ this way, it turns out that it is
just a notational variant of 
the propositional bimodal logic over 2D products of linear orders and difference frames 
(Prop.~\ref{p:foprod}).

Propositional multimodal languages interpreted in various product-like structures show up in many other contexts, and connected to several other multi-dimensional logical formalisms, such as modal and temporal description logics, and spatio-temporal logics (see \cite{gkwz03,Kurucz07} for surveys and references).
The product construction as a general combination method on modal logics was introduced in 
\cite{Segerberg73,Shehtman78,Gabbay&Shehtman98}, and has been extensively studied ever since.


\paragraph{Our contribution}
We study the \emph{satisfiability problem} of our logics in the propositional bimodal setting.
We show that satisfiability over many classes of bimodal frames with commuting 
linear and difference relations are undecidable (Theorems~\ref{t:finite}, \ref{t:mainirr}),
sometimes not even recursively enumerable (Theorems~\ref{t:omega}, \ref{t:disc}).
As a by-product, we also obtain 
new examples of finitely axiomatisable but \emph{Kripke incomplete} bimodal logics (Cor.~\ref{co:notcomplete}). 
It is easy to see (Prop.~\ref{p:cd}) 
that satisfiability over \emph{decreasing} or \emph{expanding} subframes of product frames is always reducible to `full rectangular' product frame-satisfiability. We show cases
when expanding frame-satisfiability is genuinely simpler than product-satisfiability 
(Theorems~\ref{t:omegaeu}, \ref{t:decefinu}), while it is still very complex 
(Theorems~\ref{t:decefin}, \ref{t:omegael}).

Our findings are in sharp contrast with the much lower complexity of bimodal logics over products of linear and universal frames:
Satisfiability over these is usually decidable with complexity between
{\sc ExpSpace} and {\sc 2ExpTime}  \cite{Hodkinsonetal03,Reynolds97}. 
In particular, we answer negatively a question of \cite{Reynolds97} by showing that the
addition of the `horizontal' difference operator to the decidable 
2D product of Priorian Temporal Logic over the class of all linear orders and $\Sfive$ 
results in an undecidable logic (Cor.~\ref{co:mainirr}).

Our lower bound results are also interesting because they seem to be proper generalisations
of similar results about modal products where both components are linear  \cite{Marx&Reynolds99,Reynolds&Z01,gkwz06,gkwz05a,KonevWZ05}. 
Satisfiability over linear and difference frames is of the same ({\sc NP}-complete) complexity, 
and so there are
reductions from `linear-satisfiability' to `difference-satisfiability' and vice versa. 
However, while we show (Section~\ref{expprodu}) 
how to `lift' some `difference to linear' reduction to the 2D level,
one cannot hope for such a lifting of a reverse `linear to difference' reduction:
Satisfiability over `difference$\times$difference' type products is decidable (being a fragment of two-variable classical first-order logic with counting), while 
`linear$\times$linear'-satisfiability is undecidable \cite{Reynolds&Z01}. 

Our undecidability proofs are quite different from most known undecidability proofs about 2D product logics with transitive components
\cite{Marx&Reynolds99,Reynolds&Z01,gkwz05a}.
Even if frames with two commuting relations (and so product frames) always have grid-like substructures,
there are two issues one needs to deal with
in order to encode grid-based complex problems into them:
\begin{itemize}
\item
to generate infinity, and
\item
somehow to `access' or `refer to' neighbouring-grid points, even when there might be further non-grid points around, there is no `next-time' operator in the language, and the
relations are transitive and/or dense and/or even `close to' universal.
\end{itemize}
Unlike previous proofs that first `diagonally encode' the $\omega\times\omega$-grid, 
and then use reductions of tiling or Turing machine problems, 
here we make direct use of the grid-like substructures in commutative frames, and obtain lower bounds by reductions of
counter (Minsky) machine problems. 
Representing counter machine runs apparently requires less control over 
neighbouring grid-points than tilings or Turing machine runs, and so this technique is possibly more versatile (see Section~\ref{proofidea} for more details).


\paragraph{Structure}
Section~\ref{defs} provides all the necessary definitions, and establishes connections between the two different formalisms. 
All results are then proved in the propositional bimodal setting. In particular, 
Section~\ref{chain} deals with the constant and decreasing domain cases over $\auf\omega,<\zu$
and finite linear orders. More general results on bimodal logics with `linear' and  `difference'
components are in Section~\ref{main}. The expanding domain cases
are treated in Section~\ref{domain}. 
Finally, in Section~\ref{disc} we discuss some related open problems.

Some of the results appeared in the extended abstract \cite{csl13}.


\section{Preliminaries}\label{defs}


\subsection{Propositional bimodal logics}\label{bipod}

Below we introduce all the necessary notions and notation. For more information on bimodal logics,
consult e.g.\
\cite{Blackburnetal01,gkwz03}.

We define \emph{bimodal formulas} by the following grammar:
\[
\phi::=\ \ \pred\mid\neg\phi\mid\phi\land\psi\mid\Dh\phi\mid\Dv\phi
\]
where $\pred$ ranges over an infinite set of propositional variables.
We use the usual abbreviations $\lor$, $\to$, $\leftrightarrow$, $\bot:=\pred\land\neg\pred$, 
$\top:=\neg\bot$, $\Box_i:=\neg\Diamond_i\neg$, and also
\[
\Diamond_i^+\phi:=\ \   \phi\lor\Diamond_i\phi,\hspace*{3cm} \Box_i^+\phi:=\ \ \phi\land\Box_i\phi,
\]
for $i=0,1$.
For any bimodal formula $\phi$, we denote by $\subf$ the set of its subformulas. 

A 2-\emph{frame} is a tuple $\F=\auf W,\Rh,\Rv\zu$ where $R_i$ are binary relations on the non-empty set $W$.
A \emph{model based on} $\F$ is a pair $\M=(\F,\nu)$, where $\nu$ is a
function mapping propositional variables to subsets of $W$. The \emph{truth relation}
$\M,w\models\phi$ is defined, for all $w\in W$, by induction on $\phi$ as follows:
\begin{itemize}
\item
$\M,w\models\pred$ iff $w\in\nu(\pred)$,
\item
$\M,w\models\neg\phi$ iff $\M,w\not\models\phi$,
$\M,w\models\phi\land\psi$ iff $\M,w\models\phi$ and $\M,w\models\psi$,
\item
$\M,w\models\Diamond_i\phi$ iff there exists  $v\in W$ such that $wR_iv$ and $\M,v\models\phi$ (for $i=0,1$).
\end{itemize}
We say that $\phi$ is \emph{satisfied in} $\M$, if there is $w\in W$ with $\M,w\models\phi$.
Given a set $\Sigma$ of bimodal formulas, we write $\M\models\Sigma$ if
we have $\M,w\models\phi$, for every $\phi\in\Sigma$ and every $w\in W$. 
We say that $\phi$ is \emph{valid in} $\F$, 
if $\M,w\models\phi$, for every model $\M$ based on $\F$ and
for every $w\in W$. 
If every formula
in a set $\Sigma$ is valid in $\F$, then we say that $\F$ is a \emph{frame for}  $\Sigma$.
We let $\Fr \Sigma$ denote the class of all frames for $\Sigma$.

A set $L$ of bimodal formulas is called a (normal) \emph{bimodal logic} (or \emph{logic}, for short)
if it contains all propositional tautologies and the formulas
$\Box_i(p\to q)\to(\Box_i p\to\Box_i q)$, for $i=0,1$,
and is closed  under the rules of Substitution, Modus Ponens and 
Necessitation $\varphi/\Box_i\varphi$, for $i=0,1$. Given a bimodal logic $L$, we will consider the following problem:

\medskip
\noindent
\underline{$L$-{\sc satisfiability:}}\ \ 
\parbox[t]{11.8cm}{Given a bimodal formula $\phi$, is there a model $\M$ such that $\M\models L$ and $\phi$ is satisfied in $\M$?}

\medskip
\noindent
For any class $\mathcal{C}$ of $2$-frames, we always
obtain a logic by taking
\[
\Log\,\mathcal{C}=\{\phi :\phi\mbox{ is a bimodal formula valid in every member of }\mathcal{C}\}.
\]
We say that $\Log\,\mathcal{C}$ is \emph{determined by} $\mathcal{C}$, and call such a logic
\emph{Kripke complete}. (We write just $\Log\,\F$ for $\Log\,\{\F\}$.)
Clearly, if $L=\Log\,\CC$, then there might exist frames for $L$ that are not in $\CC$, but
$L$-satisfiability is the same as the following problem:

\medskip
\noindent
\underline{$\CC$-{\sc satisfiability:}}\ \ 
\parbox[t]{11.85cm}{Given a bimodal formula $\phi$, is there a 2-frame $\F\in\CC$ such that $\phi$ is satisfied in a model based on $\F$?}

\paragraph{Commutators and products}
We might regard bimodal logics as `combinations' of their unimodal%
\footnote{Syntax and semantics of \emph{unimodal} logics are defined similarly to bimodal ones, using only one of the two modal operators. Throughout, 1-frames will be called simply \emph{frames\/}.}
 `components'.
Let $L_0$ and $L_1$ be two unimodal logics formulated using the same propositional variables and Booleans,
but having different modal operators ($\Dh$ for $L_0$ and $\Dv$ for $L_1$). Their \emph{fusion}
$\fusion$ is the smallest bimodal logic that contains both $L_0$ and $L_1$. The \emph{commutator} $\commut$
of $L_0$ and $L_1$ is the smallest bimodal logic that contains $\fusion$ and the formulas 
\begin{equation}\label{commut}
\Bv\Bh \pred\to \Bh\Bv \pred,
\qquad \Bh\Bv \pred\to \Bv\Bh \pred,
\qquad \Dh\Bv \pred\to \Bv\Dh \pred.
\end{equation}
Commutators are introduced in \cite{Gabbay&Shehtman98}, where it is also shown
that a 2-frame $\auf W,\Rh,\Rv\zu$ validates the formulas \eqref{commut} iff
\begin{itemize}
\item
$\Rh$ and $\Rv$ \emph{commute\/}:
$\forall x,y,z\,\bigl(x\Rh y\Rv z\to\exists u\,(x\Rv u\Rh z)\bigr)$, and
\item
$\Rh$ and $\Rv$ are \emph{confluent\/}:
$\forall x,y,z\,\bigl(x\Rh y\land x\Rv z\to\exists u\,(y\Rv u\land z\Rh u)\bigr)$.
\end{itemize}
Note that if at least one of $\Rh$ or $\Rv$ is symmetric, then confluence follows from commutativity.

Next, we introduce some special `two-dimensional' $2$-frames for commutators.
Given frames 
$\F_0=\auf W_0,R_0\zu$ and $\F_1=\auf W_1,R_1\zu$, their \emph{product} is defined to be
the $2$-frame
\[
\F_0\mprod\F_1= \auf W_0\mprod W_1,\Rhp,\Rvp\zu,
\]
where $W_0\mprod W_1$ is the Cartesian product of $W_0$ and $W_1$
and, for all $u,u'\in W_0$, $v,v'\in W_1$,
\begin{gather*}
\auf u,v\zu \Rhp \auf u',v'\zu\quad \text{ iff }\quad uR_0u'\mbox{ and }v=v',\\
\auf u,v\zu \Rvp \auf u',v'\zu\quad \text{ iff }\quad vR_1v' \mbox{ and }u=u'.
\end{gather*}
$2$-frames of this form will be called \emph{product frames} throughout.
For classes $\mathcal{C}_0$ and $\mathcal{C}_1$ of unimodal frames, we define
\[
\mathcal{C}_0\mprod\mathcal{C}_1=\{\F_0\mprod \F_1 : \F_i\in\mathcal{C}_i,\mbox{ for $i=0,1$}\}.
\]
Now, for $i=0,1$, let $L_i$ be a Kripke complete unimodal logic in the language with $\Diamond_i$.
The \emph{product} of $L_0$ and $L_1$ is defined as the (Kripke complete) bimodal logic
\[
L_0\times L_1 =\Log\,(\Fr L_0\mprod\Fr L_1).
\]
Product frames always validate the formulas in \eqref{commut},
and so it is not hard to see that $\commut\subseteq L_0\mprod L_1$ always holds. 
If both $L_0$ and $L_1$ are Horn axiomatisable, then $\commut=L_0\mprod L_1$ \cite{Gabbay&Shehtman98}. 
In general, $\commut$ can not only be properly contained in $L_0\mprod L_1$, but there might even be infinitely many logics in between  \cite{km12,hk14}. 

The following result of  Gabbay and Shehtman
 \cite{Gabbay&Shehtman98} is one of the few general `transfer' results on the satisfiability problem of 2D logics. It is an easy consequence of the recursive enumerability of the consequence
relation of classical (many-sorted) first-order logic:

\begin{theorem}\label{t:re}
If $\CC_0$ and $\CC_1$ are classes of frames such that both 
are recursively first-order definable in the language having a binary predicate symbol, then 
$\CC_0\mprod\CC_1$-satisfiability is co-r.e., that is, its
complement is recursively enumerable.
\end{theorem}


\paragraph{Expanding and decreasing 2-frames}
Product frames are special cases of the following construction for getting 2D frames.
Take a  (`horizontal') frame $\F=\auf W,R\zu$ and a sequence 
$\overline{\G}=\bigl\auf \G_u=\auf W_u,R_u\zu : u\in W\bigr\zu$ of 
(`vertical') frames. We can define a 2-frame by taking
\[
\Hh_{\F,\overline{\G}}=\bigl\auf\{\auf u,v\zu : u\in W,\,v\in W_u\},\Rhp,\Rvp\bigr\zu,
\]
where
\begin{align*}
\auf u,v\zu \Rhp \auf u',v'\zu\quad &\text{ iff }\quad uRu'\mbox{ and }v=v',\\
\auf u,v\zu \Rvp \auf u',v'\zu\quad &\text{ iff }\quad vR_u v' \mbox{ and }u=u'.
\end{align*}
Clearly, if $\G_x=\G_y=\G$ for all $x,y$ in $\F$, then $\Hh_{\F,\overline{\G}}=\F\mprod\G$.
However,
we can put slightly milder assumptions on the $\G_x$. We call a 2-frame of the form $\Hh_{\F,\overline{\G}}$
\begin{itemize}
\item
an \emph{expanding 2-frame} if $\G_x$ is a subframe%
\footnote{$\auf W,R\zu$ is called a \emph{subframe} of $\auf U,S\zu$, if $W\subseteq U$ and $R=S\cap (W\mprod W)$.}
of $\G_y$ whenever $xRy$, and
\item
a \emph{decreasing 2-frame} if $\G_y$ is a subframe of $\G_x$ whenever $xRy$.
\end{itemize}
So product frames are both expanding and decreasing 2-frames.
Expanding 2-frames always validate  $\Bh\Bv \pred\to \Bv\Bh \pred$ and  $\Dh\Bv \pred\to \Bv\Dh \pred$
(but not necessarily $\Bv\Bh \pred\to \Bh\Bv \pred$), and decreasing 2-frames validate $\Bv\Bh \pred\to \Bh\Bv \pred$ (but not
necessarily the other two formulas in \eqref{commut}).

For classes $\mathcal{C}_0$ and $\mathcal{C}_1$ of frames, we define
\begin{align*}
\CC_0\eprod\CC_1& =\{\mbox{expanding 2-frame }\Hh_{\F,\overline{\G}}: \F\in\CC_0,\ \G_x\in\CC_1
\mbox{ for all $x$ in $\F$}\},\\
\CC_0\dprod\CC_1& =\{\mbox{decreasing 2-frame }\Hh_{\F,\overline{\G}}: \F\in\CC_0,\ \G_x\in\CC_1
\mbox{ for all $x$ in $\F$}\}.
\end{align*}
It is not hard to see that for all classes $\mathcal{C}_0$, $\CC_1$ of frames, both
$\mathcal{C}_0\dprod\CC_1$-satisfiability and \mbox{$\mathcal{C}_0\eprod\CC_1$}-satisfiability
is reducible to $\mathcal{C}_0\mprod\CC_1$-satisfiability.
Indeed, take a fresh propositional variable $\fodompred$ (for \emph{domain}),
and for every bimodal formula $\phi$, define $\phi^{{\sf D}}$  by relativising each occurrence of $\Dh$ and $\Dv$
in $\phi$ to $\fodompred$.
Let $n$ be the nesting depth of the modal operators in $\phi$, any for any formula $\psi$ and $i=0,1$, let
\[
\Box_i^{\leq n}\psi:=\bigwedge_{k\leq n}\overbrace{\,\Box_i\dots\Box_i\phantom{I}\!\!}^k \psi.
\]
Then we have (cf.\ \cite[Thm.9.12]{gkwz03}):

\begin{proposition}\label{p:cd}\
\begin{itemize}
\item
$\phi$ is $\CC_0\dprod\CC_1$-satisfiable iff 
$\fodompred\land\Bh^{\leq n}\Bv^{\leq n}\bigl(\Dh\fodompred\to\fodompred\bigr)\land\phi^{{\sf D}}$ is $\mathcal{C}_0\mprod\CC_1$-satisfiable.

\item
$\phi$ is $\mathcal{C}_0\eprod\CC_1$-satisfiable  iff
$\fodompred\land\Bh^{\leq n}\Bv^{\leq n}\bigl(\fodompred\to\Bh\fodompred\bigr)\land\phi^{{\sf D}}$ is $\mathcal{C}_0\mprod\CC_1$-satisfiable.
\end{itemize}
\end{proposition}


\paragraph{`Linear' and `difference' logics}
Throughout, a frame $\auf W,R\zu$ is called \emph{rooted} with root $r\in W$ if every $w\in W$
can be reached from $r$ by taking finitely many $R$-steps.
By a \emph{linear order} we mean an irreflexive%
\footnote{This is just for simplifying the overall presentation. Reflexive cases are covered in Section~\ref{dense}.}%
, transitive and trichotomous relation. 
Let $\clin$ and $\clinf$ denote 
the classes of all linear orders and all finite linear orders, respectively. We let $\Kfourt:=\Log\,\clin$, that is,
the unimodal logic determined by all linear orders. 
$\Kfourt$ is well-studied as a temporal logic, and it is well-known that
frames for $\Kfourt$ are \emph{\preorders\!\!}.%
\footnote{A relation $R$ is called a \emph{\preorder\!} if it is transitive and
\emph{weakly connected\/}:
$\forall x,y,z\,\bigl(xRy\land xR z\to (\mbox{$y=z$}\lor yRz\lor zRy)\bigr)$. In other words, a rooted 
\preorder is a linear chain of clusters of universally connected points.}
A linear order $\auf W,R\zu$ is a called a \emph{well-order} if every non-empty subset of $W$ has an $R$-least element.

We denote by $\cdiff$ ($\cdifff$) the class of all (finite) \emph{difference frames}, that is, frames of the form
$\auf W,\ne\zu$ where $\ne$ is the inequality relation on $W$. We let
$\Diff:=\Log\,\cdiff$, that is, the unimodal logic determined by all difference frames. From the axiomatisation of $\Diff$ by 
Segerberg \cite{Segerberg80} it follows that frames for $\Diff$ are \emph{pseudo-equivalence}%
\footnote{A relation $R$ is called a \emph{pseudo-equivalence} if it is symmetric and \emph{pseudo-transitive\/}:
$\forall x,y,z\,\bigl(xRyRz\to (\mbox{$x=z$}\lor xRz)\bigr)$.
So a pseudo-equivalence is almost an equivalence relation, just it might have both reflexive and irreflexive points.}
relations.
If $\M$ is a model based on a rooted pseudo-equivalence frame, then we can express the uniqueness of a modally definable property in $\M$. For any formula $\phi$,
\[
\E\phi:= \ \ \D^+(\phi\land\B\neg\phi).
\]
Then, $\E\phi$ is satisfied in $\M$ iff there is a unique $w$ with $\M,w\models\phi$.

As all the axioms of $\Kfourt$ and $\Diff$, and the formulas in \eqref{commut} are Sahlqvist formulas,
the commutator $[\Kfourt,\Diff]$ is Sahlqvist axiomatisable, and so Kripke complete. Also,
\begin{multline}
\Fr[\Kfourt,\Diff]=\{\auf W,\Rh,\Rv\zu : \mbox{$\Rh$ is a \preorder\!\!,}\\
\label{commcomp}
\mbox{$\Rv$ is a pseudo-equivalence, $\Rh$ and $\Rv$ commute}\}
\end{multline}
(for more information on Sahlqvist formulas and canonicity, consult e.g.\
\cite{Blackburnetal01,cz}).


\subsection{One-variable first-order linear temporal logic with counting to two}\label{foltl}

We define \logic-\emph{formulas} by the following grammar:
\[
\phi::=\ \ \pred(x)\mid\neg\phi\mid\phi\land\psi\mid\Dt\phi\mid\existsd\,\phi
\]
where (with a slight abuse of notation) $\pred$ ranges over an infinite set $\mathcal{P}$ of \emph{monadic} predicate symbols.

A \emph{\model} is a tuple $\mathfrak M=\bigl\langle\langle T,<\rangle,D_t,I\bigr\rangle_{t\in T}$,
where $\langle T,<\rangle$ is a 
linear order, representing the \emph{timeline}, 
$D_t$ is a non-empty set, the \emph{domain at moment} $t$, for each $t\in T$, and $I$ is a function associating with every 
$t\in T$ a first-order structure $I(t)=\langle D_t,\pred^{I(t)}\rangle_{\pred\in\mathcal{P}}$.
We say that $\mathfrak M$ is \emph{based on} the linear order $\langle T,<\rangle$.
$\mathfrak M$ is a \emph{constant} (resp.\ \emph{decreasing}, \emph{expanding}) \emph{domain model\/}, if 
$D_t=D_{t'}$, (resp.\ $D_t\supseteq D_{t'}$, $D_t\subseteq D_{t'}$) 
whenever $t,t'\in T$ and $t<t'$.
A constant domain
model is clearly both a decreasing and expanding domain model as well, and can be represented as a triple $\bigl\langle\langle T,<\rangle,D,I\bigr\rangle$.

The \emph{truth-relation} $(\mathfrak M,t)\models^a\phi$ (or simply $t\models^a\phi$ if $\mathfrak M$ is understood) is
defined, for all $t\in T$ and $a\in D_t$, by induction on $\phi$ as follows:
\begin{itemize}
\item
$t\models^a\pred(x)$ iff $a\in\pred^{I(t)}$,
$t\models^a\neg\phi$ iff $t\not\models^a\phi$,
$t\models^a\phi\land\psi$ iff $t\models^a\phi$ and $t\models^a\psi$,
\item
$t\models^a\existsd\phi$ iff there exists  $b\in D_t$ such that $b\ne a$ and $t\models^b\phi$,
\item
$t\models^a\Dt\phi$ iff there is $t'\in T$ such that $t'>t$, $a\in D_{t'}$ and $t'\models^a\phi$.
\end{itemize}
We say that $\phi$ is \emph{satisfiable} \emph{in} $\mathfrak M$ if $\M,t\models^a\phi$ holds for 
some
$t\in T$ and $a\in D_t$. 
Given a class $\mathcal{C}$ of linear orders,  we say that $\phi$ is 
\logic-\emph{satisfiable in constant} (\emph{decreasing}, \emph{expanding}) \emph{domain models over} $\mathcal{C}$,
if $\phi$ is satisfiable in some constant (decreasing, expanding) domain \model\ based on some linear order from $\mathcal{C}$.

We introduce the following abbreviations:
\[
\existsx\phi:= \ \ \phi\lor\existsd\phi,\hspace*{3.5cm}
\exists^{\geq 2}x\,\phi:=\ \  \existsx(\phi\land\existsd\phi).
\]
It is straightforward to see that they have the intended semantics:
\begin{itemize}
\item
$t\models^a\existsx\phi$ iff there exists  $b\in D_t$ with $t\models^b\phi$,
\item
$t\models^a\exists^{\geq 2}x\,\phi$ iff there exist $b,b'\in D_t$ with $b\ne b'$, $t\models^b\phi$ and $t\models^{b'}\phi$.
\end{itemize}
Also, we could have chosen $\existsx$ and $\exists^{\geq 2}x$ as our primary connectives instead of $\existsd$, as 
\[
\existsd\phi\ \leftrightarrow\ (\neg\phi\land\existsx\phi)\lor\exists^{\geq 2}x\,\phi.
\]


\subsection{Connections between propositional bimodal logic and \logic}\label{conn}

Clearly, one can define a bijection
${}^\star$ from \logic-formulas to bimodal formulas, mapping each $\pred(x)$ to $\pred$,
$\Dt\phi$ to $\Dh\phi^\star$, $\existsd\phi$ to $\Dv\phi^\star$, and commuting with the Booleans.
Also, there is a bijection ${}^\dagger$ between constant domain
\model s $\M=\bigl\auf\auf T,<\zu,D,I\bigr\zu$ and modal models
$\M^\dagger=\auf\F,\nu\zu$ where $\F=\auf T,<\zu\mprod\auf D,\ne\zu$ and 
\mbox{$\nu(\pred)=\{\auf t,a\zu : \M,t\models^a\pred(x)\}$.}
Similarly, there is a one-to-one connection between expanding (decreasing) 2-frames with linear
`horizontal' and difference `vertical' components, and expanding (decreasing) domain \model s.
So it is straightforward to see the following:

\begin{proposition}\label{p:foprod}
For any class $\mathcal{C}$ of linear orders, and any \logic-formula $\phi$,
\begin{itemize}
\item
$\phi$ is \logic-satisfiable in constant domain models over $\mathcal{C}$ iff
$\phi^\star$ is $\mathcal{C}\mprod\cdiff$-satisfiable;
\item
$\phi$ is \logic-satisfiable in expanding domain models over $\mathcal{C}$ iff
$\phi^\star$ is $\mathcal{C}\eprod\cdiff$-satisfiable;
\item
$\phi$ is \logic-satisfiable in decreasing domain models over $\mathcal{C}$ iff
$\phi^\star$ is $\mathcal{C}\dprod\cdiff$-satisfiable.
\end{itemize}
\end{proposition}


\subsection{Counter machines}\label{cm}

A \emph{Minsky} or  \emph{counter machine} $M$ is described by a finite set $Q$ of states, 
a set $H\subseteq Q$ of terminal states, 
a finite set $C=\{c_0,\dots,c_{N-1}\}$ of counters with $N>1$,
a finite nonempty set  $I_q\subseteq \textit{Op}_C\times Q$ of instructions, for each $q\in Q-H$, where
each operation in $\textit{Op}_C$ is one of the following forms, for some $i<N$:
\begin{itemize}
\item
$\finci$ (\emph{increment counter} $c_i$ \emph{by one}),
\item
$\fdeci$ (\emph{decrement counter} $c_i$ \emph{by one}),
\item
$\ftest$ (\emph{test whether counter} $c_i$ \emph{is zero}).
\end{itemize}
A \emph{configuration} of $M$ is a tuple $\langle q,{\bf c}\rangle$ with $q\in Q$ representing the current state, and 
an $N$-tuple
${\bf c}=\langle c_0,\dots,c_{N-1}\rangle$ of natural numbers representing the current contents of the counters. 
For each $\iota\in \textit{Op}_C$,
we say that there is a (\emph{reliable}) $\iota$-\emph{step} between 
configurations $\sigma=\langle q,{\bf c}\rangle$ and $\sigma'=\langle q',{\bf c}'\rangle$ (written
$\sigma\stepi\sigma'$) iff there is $\langle\iota,q'\rangle\in I_q$ such that
\begin{itemize}
\item
either $\iota=\finci$ and $c_i'=c_i+1$, $c_j'=c_j$ for $j\ne i$, $j<N$,
\item
or $\iota=\fdeci$ and $c_i>0$, $c_i'=c_i-1$, $c_j'=c_j$ for $j\ne i$, $j<N$, 
\item
or $\iota=\ftest$ and $c_i'=c_i=0$, $c_j'=c_j$ for $j<N$.
\end{itemize}
We write $\sigma\step\sigma'$ iff $\sigma\stepi\sigma'$ for some $\iota\in \textit{Op}_C$.
For each $\iota\in \textit{Op}_C$,
we write $\sigma\lstepi\sigma'$ if there are configurations
$\sigma^1=\langle q,{\bf c}^1\rangle$ and $\sigma^2=\langle q',{\bf c}^2\rangle$ such that
$\sigma^1\stepi\sigma^2$, $c_i\geq c_i^1$ and $c_i^2\geq c_i'$ for every $i<N$. 
We write $\sigma\lstep\sigma'$ iff $\sigma\lstepi\sigma'$ for some $\iota\in \textit{Op}_C$.
A sequence $\langle\sigma_n : n<B\rangle$ of configurations, with $0<B\leq\omega$,
is called a \emph{run} (resp.\ \emph{lossy run}), if $\sigma_{n-1}\step\sigma_n$ (resp.\ $\sigma_{n-1}\lstep\sigma_n$)
holds for every $0<n<B$.

\medskip
Below we list the counter machine problems we will use in our lower bound proofs.

\smallskip
\noindent
\underline{{\sc CM non-termination:}}\ \ 
\parbox[t]{10.8cm}{($\Pi_1^0$-hard \cite{Minsky67})\\[3pt]
Given a counter machine $M$ and a state $q_0$, does $M$ have an infinite run starting with
$\auf q_0,{\bf 0}\zu$?}

\smallskip
\noindent
\underline{{\sc CM reachability:}}\ \ 
\parbox[t]{11.4cm}{($\Sigma_1^0$-hard \cite{Minsky67})\\[3pt]
Given a counter machine $M$, a configuration
$\sigma_0=\auf q_0,{\bf 0}\zu$ and a state $q_r$, does $M$ have a run starting with
$\sigma_0$ and reaching $q_r$?}

 \smallskip
\noindent
\underline{\sc CM recurrence:}\ \
\parbox[t]{11.6cm}{($\Sigma_1^1$-hard \cite{Alur&Henzinger})\\[3pt]
Given a counter machine $M$ and two states $q_0$, $q_r$, does $M$ have
a run starting with $\auf q_0,{\bf 0}\zu$ and visiting $q_r$ infinitely often?}

\smallskip
\noindent
\underline{{\sc LCM reachability:}}\ \ 
\parbox[t]{11cm}{(Ackermann-hard \cite{sch2})\\[3pt]
Given a counter machine $M$, a configuration
$\sigma_0=\auf q_0,{\bf 0}\zu$ and a state $q_r$, does $M$ have a lossy run starting with
$\sigma_0$ and reaching $q_r$?}

\medskip
\noindent
The Ackermann-hardness of this problem is shown by Schnoebelen \cite{sch2} without the restriction that $\sigma_0$ has all-0 counters.
It is not hard to see that this restriction does not matter: For every $M$ and $\sigma_0$
one can define a machine $M^{\sigma_0}$ that first performs incrementation steps filling the counters up to their
`$\sigma_0$-level', and then
performs $M$'s actions. Then $M$ has a lossy run starting with $\sigma_0$ and 
reaching $q_r$ iff
$M^{\sigma_0}$ has a lossy run starting with all-0 counters and 
reaching $q_r$.

 \smallskip
\noindent
\underline{{\sc LCM $\omega$-reachability:}}\ \
\parbox[t]{10.9cm}{($\Pi_1^0$-hard \cite{KonevWZ05,mayr00,sch1})\\[3pt]
Given a counter machine $M$, a configuration
$\sigma_0=\langle q_0,{\bf 0}\rangle$ and a state $q_r$,  is it the case that for every $n<\omega$
$M$ has a lossy run starting with $\sigma_0$ and visiting $q_r$ at least $n$ times?}


\subsection{Representing counter machine runs in our logics}\label{proofidea}

Before stating and proving our results, here we give a short informal guide on how we intend to use counter machines
in the various lower bound proofs of the paper. To begin with, using two different propositional variables $\state$ (for \emph{state}) and $\diag$ (for \emph{next}), we force a `diagonal staircase'  with the following properties:
\begin{enumerate}
\item[{\rm (}i{\rm )}]
every $\state$-point `vertically' ($\Rv$) sees some
$\diag$-point, and
\item[{\rm (}ii{\rm )}]
every $\diag$-point has an $\state$-point as its `immediate horizontal ($\Rh$) successor'.
\end{enumerate}
This way we not only force infinity,
but also get a `horizontal' next-time operator:
\[
\nexttime\phi:=\ \ \Bv\bigl(\diag\to\Bh(\state\to \phi)\bigr)
\]
(see Fig.~\ref{f:forward}).
In the simplest case of product frames of the form $\auf \omega,<\zu\mprod\auf W,\ne\zu$,
a grid-like structure with subsequent columns comes by definition, so everything is ready
for encoding counter machine runs in them:
Subsequent states of a run will be represented
by subsequently generated $\state$-points, and the content of each counter $c_i$ at step $n$ of a 
run will be represented by the number of $\cou$-points at the $n$th column of the grid, 
for some formula $\cou$ (see Fig.~\ref{f:forward}). As in difference frames uniqueness of a property is modally expressible, we can faithfully express the subsequent changes of the counters (see Section~\ref{chain}).

\begin{figure}[ht]
\begin{center}
\setlength{\unitlength}{.075cm}
\begin{picture}(160,103)(-5,-5)
\thinlines
\multiput(0,0)(20,0){8}{\line(0,1){85}}
\multiput(0,0)(20,0){7}{\vector(1,0){20}}
\put(140,0){\vector(1,0){7}}

\put(125,35){\line(0,1){40}}
\multiput(123,35)(0,40){2}{\line(1,0){2}}

\put(125.5,55){${}^{c_i(n)=3}$}
\multiput(114,32)(20,0){2}{$\cou$}
\multiput(73.5,43)(20,0){4}{$\cou$}
\multiput(14,47)(20,0){3}{$\cou$}
\multiput(34,77)(20,0){5}{$\cou$}
\put(53,5){$\cou$}

\thicklines
\multiput(0,0)(20,10){8}{\line(0,1){10}}
\multiput(0.25,0)(20,10){8}{\line(0,1){10}}
\multiput(0,10)(20,10){7}{\vector(1,0){19}}
\multiput(0,0)(20,10){8}{\circle*{2}}
\multiput(0,10)(20,10){8}{\circle*{2}}
\put(140,80){\vector(1,0){7}}

\put(20,45){\circle*{3.5}}
\put(40,45){\circle*{3.5}}
\put(40,75){\circle*{3.5}}
\put(60,5){\circle*{3.5}}
\put(60,45){\circle*{3.5}}
\put(60,75){\circle*{3.5}}
\put(80,75){\circle*{3.5}}
\put(80,45){\circle*{3.5}}
\put(100,75){\circle*{3.5}}
\put(100,45){\circle*{3.5}}
\put(120,75){\circle*{3.5}}
\put(120,35){\circle*{3.5}}
\put(120,45){\circle*{3.5}}
\put(140,35){\circle*{3.5}}
\put(140,45){\circle*{3.5}}

\multiput(-5,10)(20,10){8}{$\diag$}
\put(-5,0){$\state$}
\multiput(17,5)(20,10){2}{$\state$}
\multiput(62,28)(20,10){3}{$\state$}
\put(142,68){$\state$}
\put(1,-4){$q_0$}
\put(22,9){$q_1$}
\put(42,19){$q_2$}
\put(115,56){$q_n$}

\multiput(150,0)(2,0){3}{\circle*{.5}}
\multiput(150,80)(2,0){3}{\circle*{.5}}
\put(150,-4){${}_{\auf \omega,<\zu}$}

\put(-5,96){${}_{\auf W,\ne\zu}$}
\multiput(0,87)(0,2){3}{\circle*{.5}}
\multiput(20,87)(0,2){3}{\circle*{.5}}
\multiput(40,87)(0,2){3}{\circle*{.5}}
\multiput(60,87)(0,2){3}{\circle*{.5}}
\multiput(80,87)(0,2){3}{\circle*{.5}}
\multiput(100,87)(0,2){3}{\circle*{.5}}
\multiput(120,87)(0,2){3}{\circle*{.5}}
\multiput(140,87)(0,2){3}{\circle*{.5}}
\end{picture}
\end{center}
\caption{Representing counter machine runs in product frames $\auf\omega,<\zu\mprod\auf D,\ne\zu$ `going forward'.}\label{f:forward}
\end{figure}

When generalising this technique to `timelines' other than $\auf\omega,<\zu$, there can be
additional difficulties. Say, (ii) above is clearly not doable over dense linear orders. Instead of working with $\Rh$-connected points, we work with
`$\Rh$-intervals' and have the `interval-analogue' of (ii): 
Every $\diag$-interval has an $\state$-interval as its `immediate $\Rh$-successor'
(see Section~\ref{dense}). 

We also
generalise our results not only to decreasing 2-frames but for more `abstract' 2-frames having
commuting \preorder and pseudo-equivance relations (see \eqref{commcomp}).
In the abstract case, we face an additional difficulty: While commutativity does force the presence of grid-points once a diagonal
staircase is present, there might be many other non-grid points in the corresponding `vertical columns', so
the control over runs becomes more complicated.
In these cases, both the diagonal staircase and counter machine runs are forced going
`backward' (see Fig.~\ref{f:backward}), 
as this way seemingly gives us greater control over the `intended' grid-points (see Section~\ref{product}).

\begin{figure}[ht]
\begin{center}
\setlength{\unitlength}{.075cm}
\begin{picture}(190,97)(-10,-10)
\thinlines
\put(0,0){\line(0,1){75}}
\multiput(30,0)(20,0){8}{\line(0,1){75}}
\multiput(30,0)(20,0){7}{\vector(1,0){20}}
\multiput(30,10)(20,0){7}{\vector(1,0){20}}
\multiput(30,20)(20,0){7}{\vector(1,0){20}}
\multiput(30,30)(20,0){7}{\vector(1,0){20}}
\multiput(30,40)(20,0){7}{\vector(1,0){20}}
\multiput(30,50)(20,0){7}{\vector(1,0){20}}
\multiput(30,60)(20,0){7}{\vector(1,0){20}}
\multiput(30,70)(20,0){7}{\vector(1,0){20}}

\multiput(23,0)(0,10){7}{\vector(1,0){6.5}}
\multiput(0,0)(0,10){8}{\vector(1,0){6.5}}
\multiput(13,0)(2,0){3}{\circle*{.5}}
\multiput(13,10)(2,0){3}{\circle*{.5}}
\multiput(13,20)(2,0){3}{\circle*{.5}}
\multiput(13,30)(2,0){3}{\circle*{.5}}
\multiput(13,40)(2,0){3}{\circle*{.5}}
\multiput(13,50)(2,0){3}{\circle*{.5}}
\multiput(13,60)(2,0){3}{\circle*{.5}}
\multiput(13,70)(2,0){3}{\circle*{.5}}

\put(45,15){\line(0,1){30}}
\put(50,10){\line(-1,1){5}}
\put(45,45){\line(1,1){5}}
\put(31.4,31){${}^{c_i(n)=3}$}

\thicklines
\multiput(150,0)(-20,10){7}{\line(0,1){10}}
\multiput(150.25,0)(-20,10){7}{\line(0,1){10}}
\multiput(150,0)(-20,10){7}{\vector(1,0){19}}
\multiput(170,0)(-20,10){8}{\circle*{2}}
\multiput(150,0)(-20,10){7}{\circle*{2}}
\put(23,70){\vector(1,0){6.5}}

\put(-5,86){${|\ R_1}$}
\multiput(0,77)(0,2){3}{\circle*{.5}}
\multiput(30,77)(0,2){3}{\circle*{.5}}
\multiput(50,77)(0,2){3}{\circle*{.5}}
\multiput(70,77)(0,2){3}{\circle*{.5}}
\multiput(90,77)(0,2){3}{\circle*{.5}}
\multiput(110,77)(0,2){3}{\circle*{.5}}
\multiput(130,77)(0,2){3}{\circle*{.5}}
\multiput(150,77)(0,2){3}{\circle*{.5}}
\multiput(170,77)(0,2){3}{\circle*{.5}}

\put(150,0){\circle*{3.5}}
\put(130,0){\circle*{3.5}}
\put(130,10){\circle*{3.5}}
\put(110,0){\circle*{3.5}}
\put(110,10){\circle*{3.5}}
\put(110,20){\circle*{3.5}}
\put(90,10){\circle*{3.5}}
\put(90,20){\circle*{3.5}}
\put(70,10){\circle*{3.5}}
\put(70,20){\circle*{3.5}}
\put(50,10){\circle*{3.5}}
\put(50,20){\circle*{3.5}}
\put(50,50){\circle*{3.5}}
\put(30,20){\circle*{3.5}}
\put(30,50){\circle*{3.5}}

\multiput(172,1)(-20,10){8}{$\state$}
\multiput(146,-5)(-20,10){5}{$\diag$}
\put(52,45){$\diag$}
\put(31,55){$\diag$}
\put(167,-4){$q_0$}
\put(145,6){$q_1$}
\put(125,16){$q_2$}
\put(45,56){$q_n$}

\put(161,-12){${\to\ R_0}$}

\end{picture}
\end{center}
\caption{Representing counter machine runs in commutative 2-frames `going backward'.}\label{f:backward}
\end{figure}
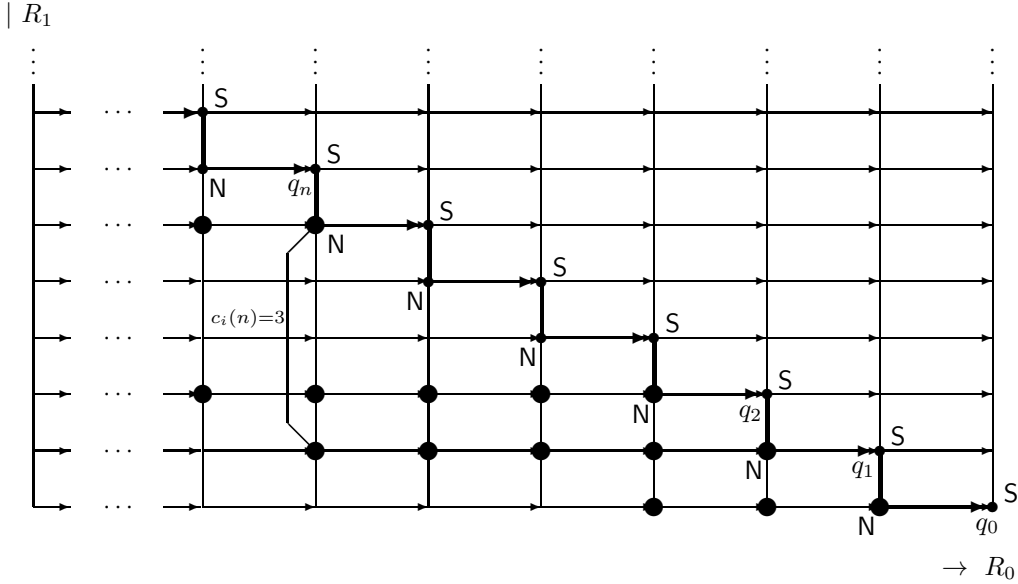

The backward technique also helps us to represent lossy counter machine runs in expanding 2-frames. When going backward horizontally in expanding 2-frames, the vertical columns might become
smaller and smaller, so some of the points carrying the information on the content of the counters might 
disappear as the runs  progress (see Section~\ref{expprodl}).


\section{$\auf\omega,<\zu$ or finite linear orders as `timelines'}\label{chain}

In this section we show the constant and decreasing domain results in the first two columns 
of Fig.~\ref{f:results}.

\begin{theorem}\label{t:omega}
$\{\auf\omega,<\zu\}\mprod\cdiff$-satisfiability is $\Sigma_1^1$-complete.
 \end{theorem}
 
\begin{theorem}\label{t:finite}
$\clinf\mprod\cdiff$-satisfiability is recursively enumerable, but undecidable.
\end{theorem}

By Prop.~\ref{p:cd}, $\CC\dprod\cdiff$-satisfiability is always reducible to $\CC\mprod\cdiff$-satisfiability.
It is not hard to see that,
whenever $\CC=\{\auf\omega,<\zu\}$ or $\CC=\clinf$, then we also have this the other way round:
$\CC\mprod\cdiff$-satisfiability is reducible to $\CC\dprod\cdiff$-satisfiability.

 \begin{proposition}\label{p:mdiscr}
 If $\CC=\{\auf\omega,<\zu\}$ or $\CC=\clinf$, then for any formula $\phi$,
 \[
 \mbox{$\phi$ is $\CC\mprod\cdiff$-satisfiable\quad iff\quad
 $\Bv^+\Bh^+(\Dh\top\to\Bv\Dh\top)\land \phi$
 is $\CC\dprod\cdiff$-satisfiable.}
\]
\end{proposition}

So by Theorems~\ref{t:omega}, \ref{t:finite} and Props.~\ref{p:foprod}, \ref{p:mdiscr} we obtain:

\begin{corollary}\label{co:undecddisc}
\logic-satisfiability recursively enumerable but undecidable in both constant decreasing domain models over the class of all finite linear orders, and $\Sigma_1^1$-complete in both constant and decreasing domain models over $\auf\mathbb \omega,<\zu$.
\end{corollary}

We prove the lower bound of Theorem~\ref{t:omega}
by reducing the `CM recurrence' problem to \mbox{$\{\auf\omega,<\zu\}$}$\mprod\cdiff$-satisfiability.
Let $\M$ be a model based on the  product of
$\auf\omega,<\zu$ and some difference frame $\auf W,\ne\zu$. First, we generate a forward going infinite 
diagonal staircase in $\M$.
Let $\diagfw$ be the conjunction of the formulas
\begin{align}
\label{initfw}
& \state\land\Bh\neg\state,\\
\label{dgenfw}
& \Bh^+\Bv^+(\state\to\Dv\diag),\\
\label{sgenfw}
& \Bh^+\Bv\bigl(\diag\to (\Dh\state\land\Bh\Bh\neg\state)\bigr).
\end{align}
%

\begin{claim}\label{c:gridfw}
Suppose that $\M,\auf 0,r\zu\models\diagfw$. Then there exists an infinite sequence
\mbox{$\auf y_m\in W : m<\omega\zu$} of points such that, for all $m<\omega$,
\begin{enumerate}
\item[{\rm (}i{\rm )}]
$y_0=r$ and for all $n<m$, $y_m\ne y_{n}$, 
\item[{\rm (}ii{\rm )}]
$\M,\auf m,y_m\zu\models\state$,
\item[{\rm (}iii{\rm )}]
if $m>0$ then $\M,\auf m-1,y_{m}\zu\models\diag$.
\end{enumerate}
\end{claim}

\begin{proof}
By induction on $m$. To begin with, $\M,\auf 0,y_0\zu\models\state$ by
\eqref{initfw}. Now suppose that for some $m<\omega$ we have
$\auf y_k : k\leq m\zu$ as required. As $\M,\auf m,y_m\zu\models\state$
by the IH, by \eqref{dgenfw} there is $y_{m+1}$ such that
$\M,\auf m,y_{m+1}\zu\models\diag$. We have $y_{m+1}\ne y_n$ for $n\leq m$
by \eqref{initfw}, \eqref{sgenfw} and the IH. Finally, 
$\M,\auf m+1,y_{m+1}\zu\models\state$ follows by \eqref{sgenfw}.
\end{proof}

Given a counter machine $M$,
we will encode runs that start with all-0 counters by going forward along the created 
diagonal staircase.
For each counter $i<N$,
we take two fresh propositional variables $\plusi$ and $\minusi$.
At each moment $n$ of time, these will be used to mark those pairs $\auf n,\dots\zu$
in $\M$ 
where $M$ increments and decrements counter $c_i$ at step $n$.
The actual content of counter $c_i$ is represented by those pairs $\auf n,\dots\zu$ where $\plusi\land\neg\minusi$ holds.
The following formula ensures that each `vertical coordinate' in $\M$ is used only once, and only previously incremented points can be decremented:
\[
\couffw:=\ \
\bigwedge_{i<N}\Bh^+\Bv^+\bigl((\plusi\to\Bh\plusi)\land
(\minusi\to\Bh\minusi)\land
(\minusi\to\plusi)\bigr).
\]
For each $i<N$, the following formulas simulate the possible changes in the counters:
\begin{align*}
\ffix & :=\ \ \Bv^+ (\Bh\plusi\to\plusi)\land\Bv^+(\Bh\minusi\to\minusi),\\
\finc & :=\ \ \Ev (\neg\plusi\land\Bh\plusi)\land\Bv^+(\Bh\minusi\to\minusi),\\
\fdec & :=\ \ \Ev (\neg\minusi\land\Bh\minusi)\land\Bv(\Bh\plusi\to\plusi).
\end{align*}
It is straightforward to prove the following:

\begin{claim}\label{c:counterfw}
Suppose that $\M,\auf 0,r\zu\models\diagfw\land\couffw$ and let, for all $m<\omega$, $i<n$,
$c_i(m) := |\{w\in W :\M,\auf m,w\zu\models\plusi\land\neg\minusi\}|$.
 Then
\[
c_i(m+1)=\left\{
\begin{array}{ll}
c_i(m), & \mbox{ if $\M,\auf m,y_{m}\zu\models\ffix$},\\[3pt]
c_i(m)+1,\ &  \mbox{ if $\M,\auf m,y_{m}\zu\models\finc$},\\[3pt]
c_i(m)-1, & \mbox{ if $\M,\auf m,y_{m}\zu\models\fdec$}.
\end{array}
\right.
\]
\end{claim}

Using the above machinery, we can encode the various counter machine instructions.
For each $\iota\in\textit{Op}_C$, we define the formula $\exei$ by taking
\[
\exei:=\ \ \left\{
\begin{array}{ll}
\displaystyle\finc\land\bigwedge_{i\ne j<N}\ffixj, & \mbox{ if $\iota=\finci$},\\
\displaystyle\fdec\land\bigwedge_{i\ne j<N}\ffixj, & \mbox{ if $\iota=\fdeci$},\\
\displaystyle\Bv^+(\plusi\to\minusi)\land\bigwedge_{j<N}\ffixj, & \mbox{ if $\iota=\ftest$}.\\
\end{array}
\right.
\]
Now we can encode runs that start with all-0 counters. 
For each $q\in Q$, we take a fresh predicate symbol $\state_q$, and 
define $\fmm$ to be the conjunction of $\couffw$ and the following formulas:
\begin{align}
\label{allzero}
& \bigwedge_{i<N}\Bv^+(\neg\plusi\land\neg\minusi),\\
 \label{griduniquetwo}
& \Bv^+\Bh^+\bigl(\state\leftrightarrow\bigvee_{q\in Q-H} (\state_q\land\bigwedge_{q\ne q'\in Q}\neg\state_{q'})\bigr),\\
\label{fwstep}
& \Bv^+\Bh^+\bigwedge_{q\in Q-H}
\Bigl[\state_q\to\bigvee_{\langle\iota,q'\rangle\in I_q}\Bigl(\exei\land\Bv\bigl(\diag\to\Bh(\state\to\state_{q'})\bigr)\Bigr)\Bigr].
\end{align}
The following lemma says that going forward along the diagonal staircase generated 
in Claim~\ref{c:gridfw}, we can force infinite recurrent runs of $M$:

\begin{lemma}\label{l:runfwd}
Suppose that $\M,\auf 0,r\zu\models\diagfw\land\fmm\land\Bh\Dh\Bv(\state\to\state_{q_r})$.
For all  $m<\omega$ and $i<N$, let 
\[
q_m:=q,\ \mbox{ if }\ \M,\auf m,y_m\zu\models\state_q,\hspace*{1cm}
c_i(m):=|\{w\in W : \M,\auf m,w\zu\models\plusi\land\neg\minusi\}|.
\]
Then $\bigl\langle\langle q_m,{\bf c}(m)\rangle : m< \omega\bigr\rangle$ is a well-defined infinite run
of $M$ starting with all-0 counters and visiting $q_r$ infinitely often.
\end{lemma}

\begin{proof}
The sequence $\langle q_m:m<\omega\rangle$ is well-defined and contains $q_r$ infinitely often by Claim~\ref{c:gridfw}(ii), \eqref{griduniquetwo} and $\Bh\Dh\Bv(\state\to\state_{q_r})$.
We show by induction on $m$ that for all $m<\omega$, 
\[
\bigl\auf \langle q_0,{\bf c}(0)\rangle,\dots,\langle q_m,{\bf c}(m)\rangle\bigr\zu
\]
 is a run of $M$ starting with all-0 counters. Indeed, $c_i(0)=0$ for $i<N$ by \eqref{allzero}.
Now suppose the statement holds for some $m<\omega$. By the IH,
$\M,\auf m,y_{m}\zu\models\state_{q_m}$. We have $q_m\in Q-H$ by Claim~\ref{c:gridfw}(ii) and \eqref{griduniquetwo}, and
so by \eqref{fwstep} there is $\langle\iota,q'\rangle\in I_{q_m}$ such that 
 $\M,\auf m,y_{m}\zu\models\exei\land\Bv\bigl(\diag\to\Bh(\state\to\state_{q'})\bigr)$.
Then $\M,\auf m+1,y_{m+1}\zu\models\state_{q'}$ by Claim~\ref{c:gridfw}.
Now there are three cases, depending on the form of $\iota$.
If $\iota=\finci$ for some $i<N$, then $c_i(m+1)=c_i(m)+1$ and $c_j(m+1)=c_j(m)$, for $j\ne i$, $j<N$,
by Claim~\ref{c:counterfw}. The case of $\iota=\fdeci$ is similar.
If $\iota=\ftest$ for some $i<N$, then  
$\M,\auf m,y_{m}\zu\models\Bv^+(\plusi\to\minusi)$, and so $c_i(m)=0$. Also, $c_j(m+1)=c_j(m)$ for
all $j<N$ by Claim~\ref{c:counterfw}. Therefore, in all cases we have
$\langle q_m,{\bf c}(m)\rangle\stepi\langle q',{\bf c}(m+1)\rangle$,  as required.
\end{proof}

On the other hand, suppose $M$ has an infinite run $\bigl\langle\langle q_m,{\bf c}(m)\rangle : m<\omega\bigl\rangle$
 starting with all-0 counters and visiting $q_r$ infinitely often.
We define a model 
$\Mrec=\mbox{$\bigl\auf\auf\omega,<\zu\mprod\auf\omega,\ne\zu,\rho\bigr\zu$}$ as follows.
For all $q\in Q$, we let
\begin{align*}
\rho(\state)& :=\{\auf n,n\zu : n<\omega\},\\
\rho(\state_q)& :=\{\auf n,n\zu : n<\omega,\ q_n=q\},\\
\rho(\diag)& :=\{\auf n,n+1\zu : n<\omega\}.
\end{align*}
Further, for all $i<N$, $n<\omega$, we define inductively the sets $\rho_n(\plusi)$ and $\rho_n(\minusi)$.
We let $\rho_0(\plusi)=\rho_0(\minusi):=\emptyset$, and
\begin{align*}
& \rho_{n+1}(\plusi):=\left\{
\begin{array}{ll}
\rho_n(\plusi)\cup\{n\}, & \mbox{ if $\iota_n=\finci$},\\
\rho_n(\plusi), & \mbox{ otherwise}.
\end{array}
\right.\\[3pt]
& \rho_{n+1}(\minusi):=\left\{
\begin{array}{ll}
\rho_n(\minusi)\cup\bigl\{\textit{min}\bigl(\rho_n(\plusi)-\rho_n(\minusi)\bigr)\bigr\}, & \mbox{ if $\iota_n=\fdeci$},\\
\rho_n(\minusi), & \mbox{ otherwise}.
\end{array}
\right.
\end{align*}
Finally, for each $i<N$, we let
\[
\label{mreclast}
\rho(\plusi) :=\{\auf m,n\zu : n\in\rho_m(\plusi)\},\hspace*{1.5cm}
\rho(\minusi) :=\{\auf m,n\zu : n\in\rho_m(\minusi)\}.
\]
It is straightforward to check that
$
\Mrec,\auf 0,0\zu\models \diagfw\land\fmm\land\Bh\Dh\Bv(\state\to\state_{q_r})
$,
showing that CM recurrence is reducible to $\auf\omega,<\zu\mprod\cdiff$-satisfiability.

As concerns the $\Sigma_1^1$ upper bound, it is not hard to see that $\auf\omega,<\zu\mprod\cdiff$-satisfiability of a bimodal formula $\phi$ is expressible by a $\Sigma_1^1$-formula over $\omega$ in the first-order language having 
binary predicate symbols $<$ and $\pred^+$, for 
each propositional variable $\pred$ in $\phi$.
This completes the proof of Theorem~\ref{t:omega}.


\bigskip
Next, we prove the lower bound of Theorem~\ref{t:finite}
by reducing the `CM reachability' problem to $\clinf\mprod\cdiff$-satisfiability.
Let $\M$ be a model based on the product of some finite linear order $\auf T,<\zu$ and some difference frame
$\auf W,\ne\zu$. We may assume that $T=|T|<\omega$.
We encode counter machine runs in $\M$ like we did in the proof of Theorem~\ref{t:omega}, but
of course this time only finite runs are possible.
We introduce a fresh propositional variable $\pend$, and 
let $\diagfwfin$ be the conjunction of \eqref{initfw}, \eqref{dgenfw} and the following version of \eqref{sgenfw}:
\begin{equation}
\label{sgenfin}
\Bh^+\Bv\bigl(\diag\land\neg\pend\to (\Dh\state\land\Bh\Bh\neg\state)\bigr).
\end{equation}
The following finitary version of Claim~\ref{c:gridfw} can be proved 
by a straightforward induction on $m$:

\begin{claim}\label{c:gridfwfin}
Suppose $\M,\auf 0,r\zu\models\diagfwfin$. Then there exist some $0<E\leq T$ and a sequence
$\auf y_m\in W: m\leq E\zu$ of points such that for all $m\leq E$,
\begin{enumerate}
\item[{\rm (}i{\rm )}]
$y_0=r$ and for all $n<m$,  $y_m\ne y_{n}$, 
\item[{\rm (}ii{\rm )}]
if $m<E$ then $\M,\auf m,y_m\zu\models\state$,
\item[{\rm (}iii{\rm )}]
if $0<m<E$ then $\M,\auf m-1,y_{m}\zu\models\diag$,
\item[{\rm (}iv{\rm )}]
$\M,\auf E-1,y_E\zu\models\pend$, and if $0<m<E-1$ then
$\M,\auf m-1,y_{m}\zu\models\neg\pend$.
\end{enumerate}
\end{claim}

The proof of the following lemma is similar to that of Lemma~\ref{l:runfwd}:

\begin{lemma}\label{l:runfwdfin}
Suppose that $\M,\auf 0,r\zu\models\diagfwfin\land\fmm\land  
\Bh^+\Bv\bigl(\diag\land\pend\to\Bv(\state\to\state_{q_r})\bigr)$.
For all  $m< E$ and $i<N$, let 
\[
q_m:=q,\ \mbox{ if }\ \M,\auf m,y_m\zu\models\state_q,\hspace*{1cm}
c_i(m):=|\{w\in W : \M,\auf m,w\zu\models\plusi\land\neg\minusi\}|.
\]
Then $\bigl\langle\langle q_m,{\bf c}(m)\rangle : m< E\bigr\rangle$ is a well-defined run
of $M$ starting with all-0 counters and reaching $q_r$.
\end{lemma}

On the other hand, suppose $M$ has a run $\bigl\langle\langle q_n,{\bf c}(n)\rangle : n< T\bigl\rangle$ for some
$T<\omega$ such that it
 starts with all-0 counters and $q_{T-1}=q_r$. 
Take the model $\Mrec$ defined in 
 the proof of Theorem~\ref{t:omega} above.
 Let $\Mfin$ be its restriction to 
 $\auf T,<\zu\mprod\auf T+1,\ne\zu$, and let
 \[
 \rho(\pend)=\{\auf T-1,T\zu\}.
 \]
Then it is straightforward to check that
\[
\Mfin,\auf 0,0\zu\models\diagfwfin\land\fmm\land
\Bh^+\Bv\bigl(\diag\land\pend\to\Bv(\state\to\state_{q_r})\bigr),
\]
completing the proof of the lower bound in Theorem~\ref{t:finite}.

As concerns the upper bound, recursively enumerability follows from the fact that $\clinf\mprod\cdiff$-satisfiability
has the `finite product model' property: 

\begin{claim}\label{c:fmpfin}
For any formula $\phi$,
if $\phi$ is $\clinf\mprod\cdiff$-satisfiable, then $\phi$ is \mbox{$\clinf\mprod\cdifff$}-satisfiable.
\end{claim}

\begin{proof}
Suppose $\M,\auf r_0,r_1\zu\models\phi$ for some model $\M$ based on the product of a finite linear
order $\auf T,<\zu$ and a (possibly infinite) difference frame $\auf W,\ne\zu$. We may assume
that $T=|T|<\omega$ and $r_0=0$.
For all $n<T$, $X\subseteq W$, we define $\cl_n(X)$ as the smallest set $Y$ such that $X\subseteq Y\subseteq W$ and having the following property:
If $x\in Y$ and $\M,\auf n,x\zu\models\Dv\psi$ for some $\psi\in\subf$, then there is $y\in Y$ such
that $y\ne x$ and
$\M,\auf n,y\zu\models\psi$.
It is not hard to see that if $X$ is finite then $\cl_n(X)$ is finite as well. In fact,
$|\cl_n(X)|\leq |X|+ 2|\subf|$.
Now let $W_0:=\cl_0(\{0\})$ and for $0<n<T$ let $W_{n}:=\cl_{n}(W_{n-1})$.
Let $\M'$ be the restriction of $\M$ to the product frame $\auf T,<\zu\mprod\auf W_{T-1},\ne\zu$.
An easy induction shows that for all $\psi\in\subf$, $n<T$, $w\in W_{T-1}$, we have
$\M,\auf n,w\zu\models\psi$ iff \mbox{$\M',\auf n,w\zu\models\psi$.}
\end{proof}


\section{Undecidable bimodal logics with a `linear' component}\label{main}

In this section we show that further combinations of \preorder and pseudo-equivalence relations 
are undecidable. First, in Subsections~\ref{product}  and \ref{discrete} we show how to represent counter machine runs in `abstract', not necessarily
product frames for commutators. Then in Subsection~\ref{dense} we extend our techniques to cover dense linear timelines. In order to obtain tighter control over the grid-structure, in all these cases we generate
both the diagonal staircase and counter machine runs going \emph{backward}, so the used formulas
force infinite \emph{rooted descending} chains in linear orders.

It is not clear, however, whether this change is always necessary, in other words, where exactly the limits of the `forward going' technique are. In particular, it would be interesting to know whether  the 
 `infinite ascending chain' analogues of the general Theorems~\ref{t:mainirr} and \ref{t:main} below hold.
 

\subsection{Between commutators and products}\label{product}

In the following theorem we do not require the bimodal logic $L$ to be Kripke complete:

\begin{theorem}\label{t:mainirr}
Let $L$ be any bimodal logic such that
\begin{itemize}
\item
$L$ contains $[\Kfourt,\Diff]$, and
\item
$\auf\omega + 1, >\zu\mprod\auf\omega,\ne\zu$ is a frame for $L$.
\end{itemize}
Then $L$-satisfiability is undecidable.
\end{theorem}

\begin{corollary}\label{co:mainirr}
Both $[\Kfourt,\Diff]$ and $\Kfourt\mprod\Diff$ are undecidable.
\end{corollary}

Note that 
Theorem~\ref{t:mainirr} is much more general than Corollary~\ref{co:mainirr}, 
as not only $[\Kfourt,\Diff]\subsetneq\mbox{$\Kfourt\mprod\Diff$}$,
but there are infinitely many different logics between them \cite{hk14}.

As a consequence of Theorems~\ref{t:re}, \ref{t:mainirr} and Prop.~\ref{p:foprod} we also obtain:

\begin{corollary}\label{co:mainfo}
\logic-satisfiability is undecidable but co-r.e.\ in constant domain models over the class of all linear orders.
\end{corollary}

We prove Theorem~\ref{t:mainirr} by
reducing `CM non-termination' to $L$-satisfiability.
 To this end, fix some model $\M$ such that $\M\models L$ and $\M$ is based on some 
 $2$-frame $\F=\auf W,\Rh,\Rv\zu$. As by our assumption $L$-satisfiability of a formula implies its
 $[\Kfourt,\Diff]$-satisfiability, by \eqref{commcomp} we may assume that $\Rh$ is transitive and weakly connected, $\Rv$ is symmetric and pseudo-transitive, and $\Rh$, $\Rv$ commute.
 We begin with forcing a \emph{unique} infinite diagonal staircase \emph{backward\/}.
Let $\diagbw$ be the conjunction of the following formulas:
\begin{align}
\label{initfbw}
& \Dh(\state\land\Bh\bot),\\
\label{dgenbw}
& \Bv^+\Dh\diag,\\
\label{sgenbw}
& \Bv^+\Bh\bigl(\diag\to(\Bv\neg\diag\land\Dv\state)\bigr),\\
\label{sgen}
& \Bv^+\Bh\bigl(\diag\to (\Dh\state\land\Bh\Bh\neg\state)\bigr),\\
\label{suni}
& \Bv^+\Bh\bigl(\state\to(\Bh\neg\state\land\Bv\neg\state)\bigr).
\end{align}
We will  show, via a series of claims,
 that $\diagbw$ forces not only a unique diagonal staircase, but also 
a unique `half-grid' in $\M$. To this end, 
for all $x\in W$, we define the \emph{horizontal rank} of $x$ by taking
\[
\hr(x):=\left\{
\begin{array}{ll}
m, & \mbox{if the length of the longest $\Rh$-path starting at $x$ is $m<\omega$},\\
\omega, & \mbox{otherwise}.
\end{array}
\right.
\]

\begin{claim}\label{c:grid}
Suppose $\M,r\models\diagbw$. Then there exist infinite sequences 
$\auf y_m : m<\omega\zu$, $\auf u_m: m<\omega\zu$, and $\auf v_m: m<\omega\zu$ of points in $W$ such that, for every $m<\omega$,
\begin{enumerate}
\item[{\rm (}i{\rm )}]
$y_m=r$ or $r\Rv y_m$, and $y_m\Rh v_m\Rh u_m$,
\item[{\rm (}ii{\rm )}]
if $m>0$ then $v_{m-1}\Rv u_m$,
\item[{\rm (}iii{\rm )}]
$\M,u_m\models\state$ and $\hr(u_m)=m$,
 \item[{\rm (}iv{\rm )}]
$\M,v_m\models\diag$ and $\hr(v_m)=m+1$.
\end{enumerate}
\end{claim}

\begin{proof}
By induction on $m$. To begin with, let $y_0=r$. By \eqref{initfbw}, there is $u_0$ such that
$y_0\Rh u_0$, $\M,u_0\models\state$ and $\hr(u_0)=0$. By \eqref{dgenbw}, there is $v_0$
such that $y_0\Rh v_0$ and $\M,v_0\models\diag$.  By \eqref{sgen}, \eqref{suni} and the weak connectedness of $\Rh$, we have that $v_0\Rh u_0$, there is no $x$ with $v_0\Rh x \Rh u_0$,
and $\hr(v_0)=1$. 

Now suppose inductively that for some $m<\omega$ we have $y_k$, $u_k$, $v_k$, for all $k\leq m$
as required.  By the IH and \eqref{sgenbw}, there is $u_{m+1}$ such that $v_m\Rv u_{m+1}$ and 
$\M,u_{m+1}\models\state$. As  $\hr(v_m)=m+1$
by the IH, we have 
$\hr(u_{m+1})=m+1$ by the commutativity of $\Rh$ and $\Rv$. 
As $y_m\Rh v_m$ by the IH, again by commutativity there is $y_{m+1}$ such that
$y_m\Rv y_{m+1}\Rh u_{m+1}$. As either $r=y_m$ or $r\Rv y_m$ by the IH and $\Rv$ is
pseudo-transitive, we have that either $r=y_{m+1}$ or $r\Rv y_{m+1}$.
So by \eqref{dgenbw}, there is $v_{m+1}$
such that $y_{m+1}\Rh v_{m+1}$ and $\M,v_{m+1}\models\diag$. 
By \eqref{sgen}, \eqref{suni} and the weak connectedness of $\Rh$, we have that $v_{m+1}\Rh u_{m+1}$, there is no $x$ with $v_{m+1}\Rh x \Rh u_{m+1}$, and $\hr(v_{m+1})=\hr(u_{m+1})+1=m+2$
as required.
\end{proof}

For each $m<\omega$, let $\C_m:=\{u_{m}\}\cup\{x\in W : x\Rv u_{m}\}$. The following claim
is a straightforward consequence of Claim~\ref{c:grid}(iii), and the commutativity of $\Rh$ and $\Rv$:

\begin{claim}\label{c:rank}
For all $m<\omega$ and all $x\in\C_m$, $\hr(x)=m$.
\end{claim}

Next, we define the half-grid points and prove some of their properties:

\begin{claim}\label{c:halfgrid}
Suppose that $\M,r\models\diagbw$. Then
for every pair $\auf m,n\zu$ with $n<m<\omega$, there exists $x_{m,n}\in\C_m$ such that 
\begin{enumerate}
\item[{\rm (}i{\rm )}]
$x_{m,m-1}=v_{m-1}$, and if $n<m-1$ then $x_{m,n}\Rh x_{m-1,n}$,
\item[{\rm (}ii{\rm )}]
if $n<m-1$ then there is no $x$ with $x_{m,n}\Rh x\Rh x_{m-1,n}$.
\end{enumerate}
Moreover, the $x_{m,n}$ are such that
\begin{enumerate}
\item[{\rm (}iii{\rm )}]
for all $x\in\C_m$, $x\Rh u_n$ iff $x=x_{m,n}$,
\item[{\rm (}iv{\rm )}]
$x_{m,n}\ne x_{m,n'}$ whever $n\ne n'$.
\end{enumerate}
\end{claim}

\begin{proof}
First, by using Claim~\ref{c:grid} throughout, we define some $x_{m,n}\in\C_m$ by induction on $m$ satisfying (i) and (ii). To begin with, let $x_{1,0}=v_1$. Now suppose that $x_{m,n}$ satisfying (i)
and (ii) have been defined 
for all $n<m$ for some $0<m<\omega$. Take any $n<m+1$. If $n=m$, then let $x_{m+1,m}=v_m$.
If $n<m$ then $v_m\Rh u_m\Rv x_{m,n}$ by the IH. So by commutativity,
there is $x_{m+1,n}$ such that $v_m\Rv x_{m+1,n}\Rh x_{m,n}$.
As $u_{m+1}\Rv v_m$, we have $x_{m+1,n}\in\C_{m+1}$ by the 
pseudo-transitivity of $\Rv$. Further, it follows from Claim~\ref{c:rank} that there is no $x$
with $x_{m+1,n}\Rh x\Rh x_{m,n}$. 

Next, we show that the $x_{m,n}$ defined above satisfy (iii) and (iv). 
As $v_n\Rh u_n$ by Claim~\ref{c:grid}(i), and $x_{m,n}\Rh v_n$ by (i),
we have $x_{m,n}\Rh u_n$ by the transitivity of $\Rh$. 
For (iii): Let $x\in \C_m$ be such that $x\Rh u_n$, and suppose that $x\ne x_{m,n}$.
Then $x\Rv x_{m,n}$, and so by commutativity, there is $z$ with 
$x_{m,n}\Rh z\Rv u_n$. As $\Rh$ is weakly connected and $\hr(u_n)=\hr(z)$ by Claim~\ref{c:rank},
we have $u_n=z$, and so $u_n\Rv u_n$ follows. As $\M,u_n\models\state$ by Claim~\ref{c:grid}(iii),
this contradicts \eqref{suni}, proving $x=x_{m,n}$.
For (iv): Suppose, for contradiction, that $x_{m,n}=x_{m,n'}$ for some $n\ne n'$. By Claim~\ref{c:grid}(iii), $hr(u_n)=n\ne n'= \hr(u_{n'})$, and so $u_n\ne u_{n'}$. As $x_{m,n}\Rh u_n$ and $x_{m,n}\Rh u_{n'}$, by the weak connectedness of $\Rh$, either $u_n\Rh u_{n'}$ or $u_{n'}\Rh u_{n}$.
As $\M,u_n\models S$ and $\M,u_{n'}\models S$ by Claim~\ref{c:grid}(iii), this contradicts
\eqref{suni}.
\end{proof}

The following claim shows that we can in fact `single out' the half-grid points in the columns by
formulas:

\begin{claim}\label{c:seesd}
Suppose that $\M,r\models\diagbw$. Then
 for all $m<\omega$ and all $x\in\C_m$,
 \begin{enumerate}
 \item[{\rm (}i{\rm )}]
 if $\M,x\models\diag$ then $m>0$ and $x=v_{m-1}=x_{m,m-1}$,
 \item[{\rm (}ii{\rm )}]
 if $\M,x\models\Dh\diag$ then $m>1$ and $x=x_{m,n}$ for some $0<n< m-1$.
 \end{enumerate}
\end{claim}

\begin{proof}
Item (i) follows from Claim~\ref{c:grid}(iv) and \eqref{sgenbw}.
For (ii):
Suppose that $\M,x\models\Dh\diag$ for some $x\in\C_m$. 
Then there is $y$ such that $x\Rh y$ and $\M,y\models\diag$. 
By Claim~\ref{c:rank}, $\hr(x)=m$, and so $\hr(y)=n$ for some $n<m$. First,
we claim that $x\ne x_{m,n}$. Indeed, suppose that $x=x_{m,n}$, Then by Claim~\ref{c:halfgrid}, either $x=v_n$ or $x\Rh v_n$. If $x=v_n$ then $\M,x\models\diag$ by Claim~\ref{c:grid}, contradicting
\eqref{sgen}. As $\hr(v_n)=n+1>n=\hr(y)$, $v_n\ne y$, the weak connectedness of$\Rh$ and $x\Rh v_n$ imply that $v_n\Rh y$, contradicting \eqref{sgen} again, and proving that $x\ne x_{m,n}$.

So we have $x\Rv x_{m,n}$. By Claim~\ref{c:halfgrid}, $x_{m,n}\Rh u_n$. So
 by commutativity there is $z$ such that $x\Rh z\Rv u_n$.
Thus, $z\in\C_n$ and so $\hr(z)=n$ by Claim~\ref{c:rank}. Then $y=z$ follows by the weak
connectedness of $\Rh$, and so $y\in \C_n$. Thus,
we have $n>0$ and $y=v_{n-1}$ by (i). Therefore, $m>1$, and $x\Rh v_{n-1}\Rh u_{n-1}$
by Claim~\ref{c:grid}. So $x=x_{m,n-1}$ follows by Claim~\ref{c:halfgrid}(iii). 
\end{proof}

Given a counter machine $M$,
we now encode runs that start with all-0 counters by going backward along the created 
diagonal staircase.
For each counter $i<N$,
we take a fresh propositional variable $\cou$. At each moment $n$ of time,
the content of counter $c_i$ at step $n$ of a run is represented by those points in $\C_n$ where $\cou$ holds.
We also force these points only to be among the half-grid points $x_{m,n}$.
We can achieve these by the following formula:
\begin{align}
\label{allf}
\couf:= &\ \ \Bv^+\Bh\bigwedge_{i<N}\bigl(\cou\to(\diag\lor\allc)\bigr),\quad\mbox{where}\\
\nonumber
\allc  := &\ \ \Dh\diag\land\Bh(\diag\lor\Dh\diag\to\cou).
\end{align}
\begin{claim}\label{c:counting}
Suppose that $\M,r\models\diagbw\land\couf$. Then for all $m<\omega$, $i<N$,
\[
|\{x\in\C_{m+1} : \M,x\models\allc\}|=|\{x\in\C_m :\M,x\models\cou\}|.
\]
\end{claim}

\begin{proof}
As $\Dh\diag$ is a conjunct of $\allc$, 
by Claims~\ref{c:halfgrid}(iv), \ref{c:seesd} and $\couf$, we have 
\begin{align*}
|\{x\in\C_{m+1} : \M,x\models\allc\}| & = |\{ n : n< m\mbox{ and }\M,x_{m+1,n}\models\allc\}|,\mbox{ and } \\
|\{x\in\C_m :\M,x\models\cou\}| & = |\{ n :n< m\mbox{ and }\M,x_{m,n}\models\cou\}|.
\end{align*}
So it is enough to show that the two sets on the right hand sides are equal. To this end,
suppose first that $ n<m$ is such that  $\M,x_{m+1,n}\models\allc$.
As $x_{m+1,n}\Rh x_{m,n}$ by Claim~\ref{c:halfgrid}(i), and $\M,x_{m,n}\models\diag\lor\Dh\diag$
by Claims~\ref{c:grid}(iv) and \ref{c:halfgrid}(i), we obtain that $\M,x_{m,n}\models\cou$.

Conversely, suppose that $\M,x_{m,n}\models\cou$ for some $n<m$.
As $n<m$,
by Claims~\ref{c:grid}(iv) and \ref{c:halfgrid}(i), we have  $\M,x_{m+1,n}\models\Dh\diag$.
Now let $x$ be such that $x_{m+1,n}\Rh x$ and $\M,x\models\diag\lor\Dh\diag$. By 
Claim~\ref{c:seesd} and the weak connectedness of $\Rh$, either $x=x_{m,n}$ or $x_{m,n}\Rh x$.
In the former case, $\M,x\models\cou$ by assumption. If $x_{m,n}\Rh x$ then 
$\M,x_{m,n}\models\neg\diag$ by \eqref{sgen}.
Therefore, $\M,x_{m,n}\models\allc$ by \eqref{allf}, and so 
$\M,x_{m,n}\models\Bh(\diag\lor\Dh\diag\to\cou)$.
 Thus, we have $\M,x\models\cou$ in this case as well, and so
$\M,x_{m+1,n}\models\Bh(\diag\lor\Dh\diag\to\cou)$ as required.
\end{proof}

Now, for each $i<N$, the following formulas simulate the possible changes that may happen in the counters when stepping backward, and also ensure
that each `vertical coordinate'  is used only once in the counting:
\begin{align}
\label{fixbw}
\ffixbw & :=\ \ \Bv^+(\cou\leftrightarrow\allc),\\
\label{incbw}
\fincbw & :=\ \ \Bv^+\bigl(\cou\leftrightarrow (\diag\lor\allc)\bigr),\\
\label{decbw}
\fdecbw & :=\ \ \Bv^+(\cou\to\allc)\land\Ev (\neg\cou\land\allc).
\end{align}
The following analogue of Claim~\ref{c:counterfw} is a straightforward consequence of Claim~\ref{c:counting}:

\begin{claim}\label{c:counter}
Suppose that $\M,r\models\diagbw\land\couf$ and let, for all $m<\omega$, $i<N$,
$c_i(m) :=|\{x\in\C_m: \M,x\models\cou\}|$. Then
\[
c_i(m+1)=\left\{
\begin{array}{ll}
c_i(m), & \mbox{ if $\M,u_{m+1}\models\ffixbw$},\\[3pt]
c_i(m)+1,\ & \mbox{ if $\M,u_{m+1}\models\fincbw$},\\[3pt]
c_i(m)-1, & \mbox{ if $\M,u_{m+1}\models\fdecbw$}.
\end{array}
\right.
\]
\end{claim}

Next, we encode the various counter machine instructions, acting backward. For each $\iota\in\textit{Op}_C$, we define the formula $\exeibw$ by taking
\[
\exeibw:=\ \ \left\{
\begin{array}{ll}
\displaystyle\fincbw\land\bigwedge_{i\ne j<N}\ffixjbw, & \mbox{ if $\iota=\finci$},\\
\displaystyle\fdecbw\land\bigwedge_{i\ne j<N}\ffixjbw, & \mbox{ if $\iota=\fdeci$},\\
\displaystyle\Bv^+\neg\cou\land\bigwedge_{j<N}\ffixjbw, & \mbox{ if $\iota=\ftest$}.\\
\end{array}
\right.
\]
Finally, we encode runs that start with all-0 counters. For each $\iota\in\textit{Op}_C$, we introduce a propositional variable $\insbw$, and
define $\fmmbwdec$ to be the conjunction of $\couf$ and
 the following formulas:
\begin{align}
 \label{gridunique}
& \Bv^+\Bh\bigl(\state\leftrightarrow\bigvee_{q\in Q-H} \bigl(\state_q\land\bigwedge_{q\ne q'\in Q}\neg\state_{q'})\bigr),\\
\label{executebwdec}
& \Bv\Bh\bigwedge_{q\in Q-H}
\bigl[\bigl(\state\land \Dv( \diag\land \Dh\state_q)\bigr) \to
\bigvee_{\langle\iota,q'\rangle\in I_q}(\insbw\land\state_{q'})\bigr],\\
\label{instbw}
& \Bv\Bh\bigwedge_{\iota\in\textit{Op}_C}(\insbw\to\exeibw).
\end{align}

The following analogue of Lemma~\ref{l:runfwd} says that going backward along the diagonal staircase generated in Claim~\ref{c:grid}, we can force infinite runs of $M$:

\begin{lemma}\label{l:runbw}
Suppose that $\M,r\models\diagbw\land\fmmbwdec$, and
for all $m<\omega$ and $i<N$ , let 
\[
 q_m  :=q, \mbox{ if }\M,u_m\models\state_{q},\quad
 c_i(m)  :=|\{x\in\C_m: \M,x\models\cou\}|,\quad
 \sigma_m  :=\langle q_m,{\bf c}(m)\rangle.
\]
Then $\langle\sigma_m:m<\omega\rangle$ is a well-defined infinite run
of $M$ starting with all-0 counters.
\end{lemma}

\begin{proof}
The sequence $\langle q_m:m<\omega\rangle$ is well-defined by Claim~\ref{c:grid}(iii) and \eqref{gridunique}.
We show by induction on $m$ that for all $m<\omega$, $\auf \sigma_0,\dots,\sigma_m\zu$ is a run of $M$ starting with all-0 counters. Indeed, $c_i(0)=0$ for $i<N$ by \eqref{allf} and Claim~\ref{c:seesd}.
Now suppose the statement holds for some $m<\omega$. 
By Claim~\ref{c:grid}, 
$\M,u_{m+1}\models\state\land \Dv( \diag\land \Dh\state_{q_m})$. 
By \eqref{gridunique} we have $q_m\in Q-H$, and
so by \eqref{executebwdec} there is 
$\langle\iota,q_{m+1}\rangle\in I_{q_m}$ such that $\M,u_{m+1}\models\insbw\land\state_{q_{m+1}}$. Therefore,
so $\M,u_{m+1}\models\exeibw$ by \eqref{instbw}. It follows from Claim~\ref{c:counter} that 
$\sigma_m\stepi\sigma_{m+1}$ as required.
\end{proof}

On the other hand,  suppose that $M$ has an infinite run $\auf \sigma_n: n<\omega\zu$ starting with all-0 counters
such that $\sigma_n=\langle q_n,{\bf c}_n\rangle$ and $\sigma_n\stepin\sigma_{n+1}$, for $n<\omega$. We define a model 
\[
\Minf=\bigl\auf\auf\omega + 1, >\zu\mprod\auf\omega,\ne\zu,\mu\bigr\zu
\]
 as follows.
For all $q\in Q$ and $\iota\in\textit{Op}_C$, we let
\begin{align}
\label{minftyfirst}
\mu(\state)& :=\{\auf n,n\zu : n<\omega\},\\
\mu(\state_q)& :=\{\auf n,n\zu : n<\omega,\ q_n=q\},\\
\mu(\diag)& :=\{\auf n+1,n\zu : n<\omega\},\\
\mu(\insbw) & :=\{\auf n,n\zu : n<\omega,\ \iota=\iota_n\}.
\end{align}
Further, for all $i<N$, $n<\omega$, we define inductively the sets $\mu_n(\cou)$.
We let $\mu_0(\cou):=\emptyset$, and
\begin{equation}\label{minfbw2}
\mu_{n+1}(\cou):=\left\{
\begin{array}{ll}
\mu_n(\cou)\cup\{n\}, & \mbox{ if $\iota_n=\finci$},\\
\mu_n(\cou)-\{\textit{min}\bigl(\mu_n(\cou)\bigr)\}, & \mbox{ if $\iota_n=\fdeci$},\\
\mu_n(\cou), & \mbox{ otherwise}.
\end{array}
\right.
\end{equation}
Finally, for each $i<N$, we let
\begin{equation}\label{minftylast}
\mu(\cou):=\{\auf m,n\zu : m<\omega,\ n\in\mu_m(\cou)\}.
\end{equation}
It is straightforward to check that
$
\Minf,\auf\omega,0\zu\models \diagbw\land\fmmbwdec
$,
showing that CM non-termination can be reduced to $L$-satisfiability. This completes the proof of Theorem~\ref{t:mainirr}.


\subsection{Modally discrete \preorders with infinite descending chains}\label{discrete}

In some cases, we can have stronger lower bounds than in Theorem~\ref{t:mainirr}.
We call a frame $\auf W,R\zu$ \emph{modally discrete} if it satisfies the following aspect of discreteness:
there are no points $x_0,x_1,\dots,x_n,\dots,x_\infty$ in $W$ such that
$x_0 R x_1 R x_2 R \dots Rx_nR\dots $, 
$x_i \ne x_{i+1}$, $x_iRx_\infty$ and $x_\infty \neg R x_i$,  for
all $i < \omega$.
We denote by $\Dis$ the logic  of all modally discrete \preorders\!\!. 
Several well-known `linear' modal logics are extensions of
$\Dis$, for example, $\Log\auf\omega,<\zu$ 
and ${\bf GL.3}$ (the logic of all Noetherian%
\footnote{$\auf W,R\zu$ is
\emph{Noetherian} if it contains no infinite ascending chains $x_0
R x_1 R x_2 R \dots$ where $x_i \ne x_{i+1}$.}
irreflexive linear orders).
Unlike `real' discreteness, modal discreteness can be captured by modal formulas, and each
of these logics is finitely axiomatisable \cite{Segerberg70,Fine85}. Also,
note that for $L\in\{\Dis,\mbox{$\Log\auf\omega,<\zu$,}{\bf GL.3}\}$, either $\auf\omega+1,>\zu$ or 
$\auf\{\infty\} \cup \mathbb Z,>\zu$ is
a frame for $L$ (here $\mathbb Z$ denotes the set of all integers).

\begin{theorem}\label{t:disc}
Let $\mathcal{C}$ be any class of frames for $[\Dis,\Diff]$ such that 
either $\auf\omega + 1, >\zu\mprod\auf\omega,\ne\zu$ or
$\auf\{\infty\} \cup \mathbb Z,>\zu\mprod\auf\omega,\ne\zu$
belongs to $\mathcal{C}$.
Then $\CC$-satisfiability is $\Sigma_1^1$-hard.
\end{theorem}

\begin{corollary}\label{co:pioneone}
Let $L_1$ be any logic from the list
\[
\Log\auf\omega,<\zu,\ {\bf GL.3},\ \Dis.
\]
Then, for any Kripke complete bimodal logic $L$ in the interval
\[
[L_1,\Diff]\ \subseteq\  L\ \subseteq\  L_1\mprod \Diff,
\]
$L$-satisfiability is $\Sigma_1^1$-hard. 
\end{corollary}

We also obtain the following interesting corollary. As $\commut$-satisfiability is clearly co-r.e whenever
 both $L_0$ and $L_1$ are finitely axiomatisable,
Corollary~\ref{co:pioneone} 
yields new examples of \emph{Kripke incomplete} commutators of Kripke complete and
finitely axiomatisable logics:

\begin{corollary}\label{co:notcomplete}
Let $L_1$ be like in Corollary~\ref{co:pioneone}. Then
the commutator $[L_1,\Diff]$ is Kripke incomplete.
\end{corollary}

Note that it is not known whether any of the commutators $[L_1,\Sfive]$ is decidable or Kripke complete, whenever $L_1$
is one of the logics in Corollary~\ref{co:pioneone}.

\bigskip
We prove Theorem~\ref{t:disc}
by reducing the `CM recurrence' problem to $\CC$-satisfiability.
Let $\M$ be a model over some $2$-frame $\F=\auf W,\Rh,\Rv\zu$ in $\CC$. 
As $\Dis\supseteq\Kfourt$, $\F$ is a frame for $\lindiffcomm$. So by \eqref{commcomp} we may assume that $\Rh$ is a modally discrete \preorder\!\!, $\Rv$ is symmetric and pseudo-transitive, and $\Rh$, $\Rv$ commute. 
We will encode counter machine runs in $\M$ `going backward', like we did in the proof of 
Theorem~\ref{t:mainirr},
with the help of the formulas $\diagbw$ and $\fmmbwdec$. This time we use
some additional machinery ensuring recurrence. 
To this end, we introduce two fresh propositional variables $\recx$ and $\predq$, and 
define the formula $\frecbw$ as the conjunction of 
the following formulas:
\begin{align}
\label{erec}
& \Bv^+\Bh(\state\to\Dv\recx),\\
\label{upd}
& \Bv^+\Bh(\recx\to\Bh\neg\state),\\
\label{dgenr}
& \Bh(\Dv\state\to\Dv\diag),\\
\label{dgenrdiff}
& \Bv\Bh\Bigl[\state\to\Bigl(\predq\leftrightarrow\Bv\bigl(\diag\to\Bh(\state\to\neg\predq)\bigr)\Bigr)\Bigr],\\
\label{rtos}
& \Bv^+\Bh(\state\land\Dh\recx\to\state_{q_r}),
\end{align}
where $q_r$ is the state of counter machine $M$ we will force to recur.
In the following claim and its proof we use the notation introduced in Claims~\ref{c:grid}--\ref{c:halfgrid}:

\begin{claim}\label{c:recbw}
Suppose that  $\M,r\models\diagbw\land\frecbw$. 
Then there are infinitely many $m$ such that $\M,u_m\models\state_{q_r}$.
\end{claim}

\begin{proof}
We show that for every $m<\omega$ there  is $k_m> m$ with $\M,u_{k_m}\models\state_{q_r}$.
Fix any $m<\omega$. By Claim~\ref{c:grid}(iii) and \eqref{erec}, there is
$w^\ast$ such that $u_m\Rv w^\ast$ and $\M,w^\ast\models\recx$. We claim that
\begin{equation}\label{recok}
\mbox{there is $k<\omega$ such that $u_{k}\Rh w^\ast$.}
\end{equation}
Indeed, suppose for contradiction that \eqref{recok} does not hold. 
We define by induction a sequence $\auf x_n : n<\omega\zu$ of points such that, for all $n<\omega$,
\begin{align}
\label{rootsees}
& r\Rh x_n,\\
\label{sqone}
& x_n\notin\C_k\mbox{ for any $k<\omega$},\\
\label{sqmid}
& \M,x_n\models\Dv\state,\\
\label{sqlast}
& \mbox{if $n>0$ then $x_{n-1}\Rh x_n$ and $x_{n-1}\ne x_n$}.
\end{align}
To begin with, by commutativity of $\Rh$ and $\Rv$, we have some $y$ with $r\Rv y\Rh w^\ast$.
So by \eqref{dgenbw}, there is $b_0$ such that $y\Rh b_0$ and $\M,b_0\models\diag$.
By \eqref{sgen}, there is $a_0$ such that $b_0\Rh a_0$, there is no $b$ with $b_0\Rh b\Rh a_0$ and
$\M,a_0\models\state$. By commutativity, there is $x_0$ such that $r\Rh x_0\Rv a_0$, and so
 $\M,x_0\models\Dv\state$.
By Claim~\ref{c:grid}(iii), \eqref{suni}, \eqref{upd} and the weak connectedness 
of $\Rh$, we have $a_0\Rh w^\ast$. Therefore, $a_0\ne u_k$ for any $k<\omega$
by our indirect assumption, and so $a_0\notin\C_k$ for any $k<\omega$ by \eqref{suni}.
As $x_0\Rv a_0$, it follows that $x_0\notin\C_k$ for any $k<\omega$.

Now suppose inductively that we have $\auf x_i : i\leq n\zu$ satisfying \eqref{rootsees}--\eqref{sqlast}
for some $n<\omega$. 
By \eqref{sqmid} of the IH and \eqref{dgenr}, there is $b_{n+1}$ such that $x_n\Rv b_{n+1}$ and
$\M,b_{n+1}\models\diag$. 
By \eqref{sgen}, there is $a_{n+1}$ such that $b_{n+1}\Rh a_{n+1}$, there is no $b$ with $b_{n+1}\Rh b\Rh a_{n+1}$ and $\M,a_{n+1}\models\state$. By commutativity, there is $x_{n+1}$ such that 
$x_n\Rh x_{n+1}\Rv a_{n+1}$, and so $r\Rh x_{n+1}$ and $\M,x_{n+1}\models\Dv\state$.
We claim that
\begin{equation}\label{xdiff}
x_{n+1}\ne x_n.
\end{equation}
Suppose for contradiction that $x_{n+1}= x_n$. Let $a_n$ be such that $x_n\Rv a_n$ and 
$\M,a_n\models\state$. Then $a_{n}=a_{n+1}$ follows by \eqref{suni}. However, by
\eqref{sgen}, \eqref{suni} and \eqref{dgenrdiff} we obtain that $a_{n}\ne a_{n+1}$. So we
have a contradiction, proving \eqref{xdiff}.
Finally, we claim that
\begin{equation}\label{aok}
x_{n+1}\notin\C_k\mbox{ for any $k<\omega$}.
\end{equation}
Suppose not, that is, $x_{n+1}\in\C_k$ for some
$k<\omega$. 
As $x_{n+1}\Rv a_{n+1}$, we also have that $a_{n+1}\in\C_k$. 
Then $\hr(b_{n+1})=k+1$, by the weak connectedness of $\Rh$ and Claim~\ref{c:rank},
and so $\hr(x_{n})=k+1$ by $x_n\Rv b_{n+1}$ and commutativity.
Take the grid-point $x_{k+1,0}\in\C_{k+1}$ defined in Claim~\ref{c:halfgrid}. As $\hr(x_{k+1,0})=k+1$ by Claim~\ref{c:rank},
we have $x_{k+1,0}=x_n$ by the weak connectedness of $\Rh$. But this contradicts \eqref{sqone} of the IH,
proving \eqref{aok}.

So we have defined $\auf x_n : n<\omega\zu$ satisfying \eqref{rootsees}--\eqref{sqlast}.
As $\hr(u_0)=0$ by Claim~\ref{c:grid}(iii), and $\M,x_n\models\neg\state$ by \eqref{suni} and \eqref{sqmid},
 by the weak connectedness of $\Rh$ we obtain that $x_n\Rh u_0$ for every $n<\omega$. This contradicts the modal discreteness of $\Rh$, and so proves \eqref{recok}.

Now let $k_m$ be such that $u_{k_m}\Rh w^\ast$. As $w^\ast\in\C_m$,
$k_m> m$ follows from Claim~\ref{c:rank}.
By Claim~\ref{c:grid}(iii) and \eqref{rtos}, we have $\M,u_{k_m}\models\state_{q_r}$ as required.
\end{proof}

Now the following lemma is a straightforward consequence of Lemma~\ref{l:runbw} and Claim~\ref{c:recbw}:

\begin{lemma}\label{l:runbwrec}
Suppose that $\M,r\models\diagbw\land\fmmbwdec\land\frecbw$, and
for all $m<\omega$, $i<N$, let 
\[
 q_m  :=q, \mbox{ if }\M,u_m\models\state_{q},\quad
 c_i(m)  :=|\{x\in\C_m: \M,x\models\cou\}|,\quad
 \sigma_m  :=\langle q_m,{\bf c}(m)\rangle.
 \]
Then $\langle\sigma_m:m<\omega\rangle$ is a well-defined run
of $M$ starting with all-0 counters and visiting $q_r$ infinitely often.
\end{lemma}

On the other hand, suppose that $M$ has run $\bigl\langle\langle q_n,{\bf c}(n)\rangle : n<\omega\bigl\rangle$
such that ${\bf c}(0)=0$ and $q_{k_n}=q_r$ for an infinite sequence $\langle k_n: n<\omega\rangle$. Clearly, we may assume that $k_n> n$, for $n<\omega$.
By assumption, $\F\in\CC$ for either $\F=\auf\omega + 1, >\zu\mprod\auf\omega,\ne\zu$ or
$\F=\mbox{$\auf\{\infty\} \cup \mathbb Z,>\zu$}\mprod\auf\omega,\ne\zu$.
Then the model $\Minf$ defined in  \eqref{minftyfirst}--\eqref{minftylast} can be regarded as a model based on
$\F$, and we may add
\[
\mu(\predq):=\{\auf n,n\zu : n<\omega,\ \mbox{$n$ is odd}\},\qquad\qquad
\mu(\recx):=\{\auf n,k_n\zu : n<\omega\}.
\]
It is straightforward to check that
$\Minf,\auf\omega,0\zu\models \diagbw\land\fmmbwdec\land\frecbw$.
So by Lemma~\ref{l:runbwrec}, CM recurrence can be reduced to 
\mbox{$\CC$-satisfiability,} proving Theorem~\ref{t:disc}.


\subsection{Decreasing 2-frames based on dense \preorders}\label{dense}

A \preorder $\auf W,R\zu$ is called \emph{dense} if 
$\forall x,y\,(xRy\to\exists z\,xRzRy)$. Well-known examples of dense linear orders are
$\auf\mathbb Q,<\zu$ and $\auf\mathbb R,<\zu$ of the \emph{rationals} and the \emph{reals}, respectively. Neither Theorem~\ref{t:omega} nor Theorem~\ref{t:mainirr} apply if the `horizontal component' of a bimodal logic has only dense frames.
In this section we cover some of these cases.

We say that a frame $\F = \auf W, R\zu$ \emph{contains an\/} $\auf\omega + 1, >\zu$-\emph{type chain\/},
if there are distinct points $x_n$, for $n\leq\omega$, in $W$ such that  $x_nRx_m$ iff $n>m$, 
for all $n,m\leq\omega$, $n\ne m$. 
Observe that this is less than saying that $\F$ has a subframe isomorphic to $\auf\omega + 1, >\zu$,
as for each $n$, $x_n R x_n$ might or might not hold. So $\F$ can be reflexive and/or dense, and still have this property.
We have the following generalisation of Theorem~\ref{t:mainirr} for classes of 
decreasing \mbox{2-frames:}

\begin{theorem}\label{t:main}
Let $\CC$ be any class of \preorders such that $\F\in\CC$ for  some $\F$ containing an $\auf\omega + 1, >\zu$-type chain. Then $\CC\dprod\cdiff$-satisfiability is undecidable.
\end{theorem}

As a consequence of Theorem~\ref{t:main} and Props.~\ref{p:cd}, \ref{p:foprod} we obtain:

 \begin{corollary}\label{co:densefo}
\logic-satisfiability is undecidable both in decreasing and in constant domain models 
over $\auf\mathbb Q,<\zu$ and over $\auf\mathbb R,<\zu$.
\end{corollary}

Also,
as a consequence of Theorems~\ref{t:re}, \ref{t:main} and Props.~\ref{p:cd}, \ref{p:foprod} we have:

\begin{corollary}\label{co:mainfod}
\logic-satisfiability is undecidable but co-r.e.\ in decreasing domain models over the class of all linear orders.
\end{corollary}

We prove Theorem~\ref{t:main} by
reducing the `CM non-termination' problem to $\CC\dprod\cdiff$-satisfiability.
We intend to use something like the formula $\diagbw\land\fmmbwdec$ 
defined in the proof of Theorem~\ref{t:mainirr}. 
The problem is that if $\auf W,R\zu$ is reflexive and/or dense, 
then a formula of the form $\Dh\state\land\Bh\Bh\neg\state$ in conjunct \eqref{sgen} of $\diagbw$ is clearly not satisfiable.
In order to overcome this, we will apply a version of the well-known `tick trick'
(see e.g.\ \cite{Spaan93,Reynolds&Z01,gkwz05a}).

So let $\M$ be a model based on a decreasing 2-frame
$\Hh_{\F,\overline{\G}}$ where $\F=\auf W,R\zu$ is a \preorder\!\!,
 and for every $x\in W$, $\G_x=\auf W_x,\ne\zu$.
We may assume that $\F$ is rooted with some $r_0$ as its root.
We take a fresh propositional variable $\tick$, and define a new modal operator by setting, for every formula $\psi$,
\begin{align*}
& \DhM\psi:=\bigl [\tick\land\Dh\bigl(\neg \tick\land (\psi\lor\Dh\psi)\bigr)\bigr]\lor
\bigl [\neg \tick\land\Dh\bigl(\tick\land (\psi\lor\Dh\psi)\bigr)\bigr],\mbox{ and }\\
& \BhM\phi:=\neg\DhM\neg\psi.
\end{align*}
Now suppose that $\M,\auf r_0,r_1\zu\models\eqref{tick}$, where
\begin{equation}\label{tick}
\Bv^+\Bh^+\bigl(\tick\lor\Dv \tick\to (\tick\land\Bv \tick)\bigr).
\end{equation}
We define a new binary relation
$\RhM$ on $W$ by taking, for all $x,y\in W$,
\begin{multline*}
x\RhM y\quad\mbox{iff}\quad \exists\, 
z\in W\ \bigl(xR z\mbox{ and }(z=y\mbox{ or }zR y)\mbox{ and }\\
\forall u\in W_z\,(\M,\auf x,u\zu\models \tick\ \leftrightarrow\ \M,\auf z,u\zu\models\neg \tick)\bigr).
\end{multline*}
Then it is not hard to check that $\RhM$ is transitive, and $\DhM$ behaves like
a `horizontal' modal diamond w.r.t. $\RhM$ in $\M$, that is, for all $x\in W$, $u\in W_x$,
\[
\M,\auf x,u\zu\models\DhM\psi\quad\mbox{iff}\quad
\exists y\in W\ \bigl(x\RhM y,\ u\in W_y\mbox{ and }\M,\auf y,u\zu\models\psi\bigr).
\]
However, $\RhM$ is not necessarily weakly connected. We only have:
\begin{equation}\label{wconM}
\forall x,y,z\,\bigl(x\RhM y\land x\RhM z\to (y\sim z\lor y\RhM z\lor z\RhM y)\bigr),
\end{equation}
where 
\[
y\sim z\qquad\mbox{iff}\qquad 
\mbox{either }y=z\ 
\mbox{ or }\ \bigl(y R z\mbox{ and }y\neg\RhM z\bigr)
\mbox{ or }\  \bigl(z R y \mbox{ and }z\neg\RhM y\bigr).
\]
The relation $\sim$ can be genuinely larger than equality. It is not hard to check (using that
$\auf W,R\zu$ is rooted) that $\sim$ is an equivalence relation, and
$\sim$-related points have the following properties:
\begin{align}
\label{sameroot}
& \forall x,y,z\ (y\sim z\land x\RhM y\to x\RhM z),\\
\label{wcon}
& \forall x,y,z\ (y\sim z\land y\RhM x\to z\RhM x).
\end{align}
We would like our propositional variables to behave `uniformly' when interpreted at pairs with $\sim$-related first components (that is, along `horizontal intervals').
To achieve this, 
for a propositional variable $\pred$, let $\pint$ denote conjunction of the following formulas:
\begin{align}
\label{puniq}
& \Bv^+\Bh^+\bigl(\pred\to\BhM\neg\pred\bigr),\\
\label{int1}
& \Bv^+\Bh^+\bigl(\Dh\pred\land\BhM\neg\pred\to\pred\bigr),\\
\label{int5}
& \Bv^+\Bh^+\bigl(\pred\land\neg\DhM\top\to\Bh\pred \bigr),\\
\label{int2}
& \Bv^+\Bh^+\bigl(\pred\land\DhM\top \to\DhM\predp \bigr),\\
\label{int4}
& \Bv^+\Bh^+\bigl(\pred\to\Bh(\DhM\predp\to\pred) \bigr),
\end{align}
where $\predp$ is a fresh propositional variable. We also introduce the following notation, for all 
$x\in W$, $y\in W_x$ and all formulas $\phi$:
\[
\M,\auf \Ix(x),y\zu\models\phi\qquad\mbox{iff}\qquad
\M,\auf z,y\zu\models\phi\ \mbox{ for all $z$ such that $z\sim x$ and $y\in W_z$}.
\]

\begin{claim}\label{c:int}
Suppose that $\M,\auf r_0,r_1\zu\models\eqref{tick}\land\pint$.
For all $x\in W$, $y\in W_x$, if $\M,\auf x,y\zu\models\pred$ then $\M,\auf \Ix(x),y\zu\models\pred$.
\end{claim}

\begin{proof}
Suppose that $\M,\auf x,y\zu\models\pred$. Take some
$z\sim x$ with $z\ne x$ and $y\in W_z$. Suppose first that $zRx$.
As $\M,\auf x,y\zu\models\BhM\neg\pred$ by \eqref{puniq}, we have $\M,\auf z,y\zu\models\BhM\neg\pred$ by \eqref{wcon}.
Therefore, $\M,\auf z,y\zu\models\pred$ by \eqref{int1}. 

Now suppose that $xRz$. There are two cases:
If $\M,\auf x,y\zu\models\neg\DhM\top$ then $\M,\auf z,y\zu\models\pred$ follows by \eqref{int5}.
If $\M,\auf x,y\zu\models\DhM\top$ then $\M,\auf x,y\zu\models\DhM\predp$ by \eqref{int2}. 
Thus, $\M,\auf z,y\zu\models\DhM\predp$ by \eqref{wcon}.
So $\M,\auf z,y\zu\models\pred $ follows by \eqref{int4}.
\end{proof}

Throughout, for any formula $\phi$, we denote by $\phi^\bullet$ the formula obtained from
$\phi$ by replacing each occurrence of $\Dh$  with $\DhM$.
Now all the necessary tools are ready for forcing a unique infinite diagonal staircase of intervals, going backward. In decreasing 2-frames this will also automatically give us an infinite half-grid.
To this end, take the formula $\diagbw$ defined in \eqref{initfbw}--\eqref{suni}.
We define a new formula $\diagbwd$ by modifying $\diagbw$ as follows. First, replace 
the conjunct \eqref{initfbw} by the slightly stronger
\begin{equation}\label{initbwd}
\DhM(\state\land\Bv^+\BhM\bot),
\end{equation}
then replace each remaining conjunct $\phi$ in $\diagbw$ by $\phi^\bullet$.
Finally, add the conjuncts \eqref{tick} and $\pint$, for $\pred\in\{\diag,\state\}$.
We then have the following analogue of Claims~\ref{c:grid}--\ref{c:halfgrid}:

\begin{claim}\label{c:gridd}
Suppose that $\M,\auf r_0,r_1\zu\models\diagbwd$. Then there exist infinite sequences
$\auf x_m\in W:m<\omega\zu$ and $\auf y_m\in W_{x_m} : m<\omega\zu$ such that for all $m<\omega$,
\begin{enumerate}
\item[{\rm (}i{\rm )}]
$y_m\ne y_{n}$, for all $n<m$,
\item[{\rm (}ii{\rm )}]
there is no $x$ with $x_0\RhM x$, and 
if $m>0$ then $x_m \RhM x_{m-1}$, and there is no $x$ such that $x_m \RhM x \RhM x_{m-1}$,
\item[{\rm (}iii{\rm )}]
$\M,\auf \Ix(x_m),y_m\zu\models\state$,
\item[{\rm (}iv{\rm )}]
if $m>0$ then $\M,\auf \Ix(x_m),y_{m-1}\zu\models\diag$.
\end{enumerate}
\end{claim}

\begin{proof}
By induction on $m$. To begin with, let $y_0=r_1$. By \eqref{initbwd}, there is $x_0$ such that
$r_0\RhM x_0$, $y_0\in W_{x_0}$, $\M,\auf x_0,y_0\zu\models\state$ and
\begin{equation}\label{deadend}
\M,\auf x_0,y_0\zu\models\Bv^+\BhM\bot.
\end{equation}
By $\sint$, we have $\M,\auf \Ix(x_0),y_0\zu\models\state$. 

Now suppose inductively that for some $m<\omega$ we have $x_k$, $y_k$, for all $k\leq m$
as required.  By the IH, $y_m\in W_{x_m}\subseteq W_{r_0}$, so by \eqref{dgenbw}$^\bullet$,
there is $x_{m+1}$ such that $r_0\RhM x_{m+1}$, $y_m\in W_{x_{m+1}}$ and 
$\M,\auf x_{m+1},y_{m}\zu\models\diag$. 
By \eqref{sgen}$^\bullet$, \eqref{suni}$^\bullet$, \eqref{wconM} and \eqref{wcon},
we have that $x_{m+1}\RhM x_{m}$, and there is no $x$ with $x_{m+1}\RhM x \RhM x_{m+1}$.
By $\dint$, we have $\M,\auf \Ix(x_{m+1}),y_m\zu\models\diag$. 
By \eqref{sgenbw}$^\bullet$, there is $y_{m+1}$ such that $y_{m+1}\ne y_m$,
$y_{m+1}\in W_{x_m}$ and 
\mbox{$\M,\auf x_{m+1},y_{m+1}\zu\models\state$.}
By $\sint$, we have $\M,\auf \Ix(x_{m+1}),y_{m+1}\zu\models\state$. 
Finally, we have $y_{m+1}\ne y_{n}$ for $n< m$ by \eqref{suni}$^\bullet$.
\end{proof}

We have the following analogue of Claim~\ref{c:seesd}:

\begin{claim}\label{c:seesdd}
Suppose that $\M,\auf r_0,r_1\zu\models\diagbwd$.
 For all $m<\omega$ and all $y\in W_{x_m}$,
 \begin{enumerate}
 \item[{\rm (}i{\rm )}]
 if there is $z$ such that $z\sim x_m$, $y\in W_z$ and 
 $\M,\auf z,y\zu\models\diag$, then $m>0$ and $y=y_{m-1}$,
 \item[{\rm (}ii{\rm )}]
 if there is $z$ such that $z\sim x_m$, $y\in W_z$ and 
 $\M,\auf z,y\zu\models\DhM\diag$, 
 then $m>1$ and $y=y_{n}$ for some $0<n< m-1$.
 \end{enumerate}
\end{claim}

\begin{proof}
For (i): 
Take some $z$ such that  $z\sim x_m$, $y\in W_z$ and  $\M,\auf z,y\zu\models\diag$.
If $m=0$, then $\M,\auf z,y\zu\models\BhM\bot$ by \eqref{wcon} and \eqref{deadend},
and so $\M,\auf z,y\zu\models\neg\diag$ by \eqref{sgen}$^\bullet$. So we may assume that $m>0$.
Then by \eqref{wcon} and Claim~\ref{c:gridd}(ii), we have $z\RhM x_{m-1}$, and so $y_{m-1}\in W_z$.
Now (i) follows from Claim~\ref{c:gridd}(iv) and \eqref{sgenbw}$^\bullet$.

For (ii): 
Take some $z$ such that  $z\sim x_m$, $y\in W_z$ and  $\M,\auf z,y\zu\models\DhM\diag$.
Then by \eqref{wcon}, there is $u$ such that $x_m\RhM u$, $y\in W_u$ and 
$\M,\auf u,y\zu\models\diag$. By Claim~\ref{c:gridd}(ii), $u\sim x_n$ for some $n<m$, and so
by (i), $y=y_{n-1}$ as required.
\end{proof}

Given a counter machine $M$,  we intend to encode its runs going backward along the
diagonal staircase of intervals, using again a propositional variable $\cou$ for each $i<N$ to represent the changing content of each counter. 
To this end, recall the formula $\couf$ defined in \eqref{allf}, and consider
\begin{align*}
\coufd:= &\ \ \Bv^+\BhM\bigwedge_{i<N}\bigl(\cou\to(\diag\lor\allcd)\bigr),\quad\mbox{where}\\
\allcd  := &\ \ \DhM\diag\land\BhM(\diag\lor\DhM\diag\to\cou).
\end{align*}

Then we have the following analogue of Claim~\ref{c:counting}:

\begin{claim}\label{c:countingd}
Suppose $\M,\auf r_0,r_1\zu\models\diagbwd\land\coufd$. Then for all $m<\omega$, $i<N$,
\[
|\{y\in W_{x_{m+1}} : \M,\auf \Ix(x_{m+1}),y\zu\models\allcd\}| =
|\{y\in W_{x_m} :\M,\auf \Ix(x_m),y\zu\models\cou\}|.
\]
\end{claim}

\begin{proof}
As $\DhM\diag$ is a conjunct of $\allcd$, 
by Claim~\ref{c:seesdd} and $\coufd$, we have 
\begin{align*}
|\{y\in W_{x_{m+1}} : \M,\auf \Ix(x_{m+1}),y\zu\models\allcd\}| &
 = |\{n: n< m\mbox{ and }\M,\auf\Ix(x_{m+1}),y_n\zu\models\allcd\}|, \\
 |\{y\in W_{x_m} :\M,\auf \Ix(x_m),y\zu\models\cou\}| &
 = |\{n  :n< m\mbox{ and }\M,\auf \Ix(x_{m}),y_n\zu\models\cou\}|.
\end{align*}
So it is enough to show that the two sets on the right hand sides are equal. 
Suppose first that $n<m$ is such that $\M,\auf\Ix(x_{m+1}),y_n\zu\models\allcd$,
and so 
\[
\M,\auf x_{m+1},y_n\zu\models\BhM(\diag\lor\DhM\diag\to\cou).
\]
Thus, in order to prove that $\M,\auf \Ix(x_{m}),y_n\zu\models\cou$,
it is enough to show that for all $z$ such that $z\sim x_m$ and $y_n\in W_z$, we have
\begin{equation}\label{needed}
x_{m+1}\RhM z\ \mbox{ and }\M,\auf z,y_n\zu\models\diag\lor\DhM\diag.
\end{equation}
To this end, we have $x_{m+1}\RhM x_m$ by Claim~\ref{c:gridd}(ii), and so $x_{m+1}\RhM z$
follows by \eqref{sameroot}.
If $n=m-1$ then $\M,\auf z,y_n\zu\models\diag$ by Claim~\ref{c:gridd}(iv).
If $n<m-1$ then $x_{m}\RhM x_{n+1}$ by Claim~\ref{c:gridd}(ii) and the transitivity of $\RhM$,
and so $z\RhM x_{n+1}$  by \eqref{sameroot}. As 
$\M,\auf x_{n+1},y_n\zu\models\diag$ by Claim~\ref{c:gridd}(iv), we obtain
$\M,\auf z,y_n\zu\models\DhM\diag$, as required in \eqref{needed}.

Conversely, suppose that $\M,\auf \Ix(x_{m}),y_n\zu\models\cou$ for some $n<m$.
As $n<m$,
by Claims~\ref{c:gridd}(ii),(iv) and \eqref{sameroot}, we have  $\M,\auf\Ix(x_{m+1}),y_n\zu\models\DhM\diag$. 
In order to prove $\M,\auf\Ix(x_{m+1}),y_n\zu\models\allcd$,
it remains to show that 
\begin{equation}\label{neededtwo}
\M,\auf\Ix(x_{m+1}),y_n\zu\models\BhM(\diag\lor\DhM\diag\to\cou).
\end{equation}
To this end, let $u,z$ be such that $u\sim x_{m+1}$, $u\RhM z$, $y_n\in W_z$ and
$\M,\auf z,y_n\zu\models\diag\lor\DhM\diag$. By \eqref{sameroot}, we have $x_{m+1}\RhM z$, and so
by \eqref{wconM} and Claim~\ref{c:gridd}(ii), either $z\sim x_{m}$ or $x_{m}\RhM z$.
In the former case, $\M,\auf z,y_n\zu\models\cou$ by assumption. If $x_{m}\RhM z$ then 
$\M,\auf x_m,y_n\zu\models\neg\diag$ by \eqref{sgen}$^\bullet$, and so 
$\M,\auf x_{m},y_n\zu\models\allcd$ by $\coufd$.
Thus, $\M,\auf x_{m},y_n\zu\models\BhM(\diag\lor\DhM\diag\to\cou)$, and so
$\M,\auf z,y_n\zu\models\cou$ follows in this case as well, proving \eqref{neededtwo}.
\end{proof}

Now recall the formulas $\ffixbw$, $\fincbw$ and $\fdecbw$ from \eqref{fixbw}--\eqref{decbw}, simulating the possible changes in the counters stepping backward, and ensuring that each `vertical coordinate'  is used only once in the counting.
Observe that $\Bv^+\Bh^+\bigl(\cou\to\BhM\neg\cou\bigr)$ (conjunct \eqref{puniq} of $\cinti$)
and $\coufd$ cannot hold simultaneously, so 
we cannot use the formula $\cinti$ for forcing $\cou$ to behave uniformly in intervals.
However, as each vertical coordinate is used at most once in the counting,
we can force that the \emph{changes} happen uniformly in the intervals (even
when the counter is decremented). To this end,
for each $i<N$ we introduce a fresh propositional variable $\minusi$, and then postulate
\begin{equation}\label{decuniq}
\bigwedge_{i<N}\Bigl(\cint\land \Bv^+\Bh\bigl(\minusi\leftrightarrow (\neg\cou\land\allcd)\bigr)\Bigr).
\end{equation}
Now we have the following analogue of Claim~\ref{c:counter}:

\begin{claim}\label{c:counterd}
Suppose that $\M,\auf r_0,r_1\zu\models\diagbwd\land\coufd\land\eqref{decuniq}$ and, for all $m<\omega$, $i<n$, let
$c_i(m) := |\{y\in W_{x_m} :\M,\auf\Ix(x_m),y\zu\models\cou\}|$.
 Then
\[
c_i(m+1)=\left\{
\begin{array}{ll}
c_i(m), & \mbox{ if $\M,\auf\Ix(x_{m+1}),y_{m+1}\zu\models\ffixbwd$},\\[3pt]
c_i(m)+1,\ &  \mbox{ if $\M,\auf\Ix(x_{m+1}),y_{m+1}\zu\models\fincbwd$},\\[3pt]
c_i(m)-1, & \mbox{ if $\M,\auf\Ix(x_{m+1}),y_{m+1}\zu\models\fdecbwd$}.
\end{array}
\right.
\]
\end{claim}

\begin{proof}
We show only the hardest case, when $\M,\auf\Ix(x_{m+1}),y_{m+1}\zu\models\fdecbwd$.
The other cases are similar and left to the reader.
As $\M,\auf x_{m+1},y_{m+1}\zu\models\Ev (\neg\cou\land\allcd)$, there is an $y^\ast\in W_{x_{m+1}}$
such that
\begin{align}
\label{echange}
& \M,\auf x_{m+1},y^\ast\zu\models\neg\cou\land\allcd,\\
\label{uchange}
& \M,\auf x_{m+1},y\zu\not\models\neg\cou\land\allcd, \mbox{for all $y\ne y^\ast$, $y\in W_{x_{m+1}}$}.
\end{align}
We claim that
%
\begin{multline}
\{y\in W_{x_{m+1}} :\M,\auf\Ix(x_{m+1}),y\zu\models\cou\}\cup\{y^\ast\}=\\
\label{decok}
\{y\in W_{x_{m+1}} :\M,\auf\Ix(x_{m+1}),y\zu\models\allcd\}.
\end{multline}
%
Indeed, in order to show the $\subseteq$ direction, suppose first that 
$\M,\auf\Ix(x_{m+1}),y\zu\models\cou$ for some $y\in W_{x_{m+1}}$.
Then by the first conjunct of $\fdecbwd$, we have $\M,\auf\Ix(x_{m+1}),y\zu\models\allcd$. 
Further, we have $\M,\auf\Ix(x_{m+1}),y^\ast\zu\models\allcd$ 
by \eqref{echange}, \eqref{decuniq} and Claim~\ref{c:int}.
For $\supseteq$, suppose that $\M,\auf\Ix(x_{m+1}),y\zu\models\allcd$
for some $y\in W_{x_{m+1}}$, $y\ne y^\ast$.
Then by \eqref{uchange}, \eqref{decuniq} and Claim~\ref{c:int}, we have 
$\M,\auf\Ix(x_{m+1}),y\zu\models\neg\minusi$, and so $\M,\auf\Ix(x_{m+1}),y\zu\models\cou$,
proving \eqref{decok}.

Now $c_i(m+1)+1=c_i(m)$ follows from \eqref{decok} and Claim~\ref{c:countingd}.
\end{proof}

Given a counter machine $M$,  recall the formula $\fmmbwdec$ defined in the proof of
Theorem~\ref{t:mainirr} (as the conjunction of \eqref{allf} and \eqref{gridunique}--\eqref{instbw}).
Let $\fmmbwdecd$ be the conjunction of $\fmmbwdecda$,
\eqref{decuniq} and  
$\pint$, for $\pred\in\{\state_q,\insbw\}_{q\in Q,\,\iota\in\textit{Op}_C}$.
Then we have the following analogue of Lemma~\ref{l:runbw}:

\begin{lemma}\label{l:runbwd}
Suppose that $\M,\auf r_0,r_1\zu\models\diagbwd\land\fmmbwdecd$, and
for all $m<\omega$, $i<N$, let 
\[
 q_m  :=q, \mbox{ if }\M,\auf \Ix(x_m),y_m\zu\models\state_{q},\qquad
 c_i(m)  :=|\{y\in W_{x_m}: \M,\auf \Ix(x_m),y\zu\models\cou\}|.
\]
Then $\bigl\langle\langle q_m,{\bf c}(m)\rangle:m<\omega\bigr\rangle$ is a well-defined infinite run
of $M$ starting with all-0 counters.
\end{lemma}

\begin{proof}
The sequence $\langle q_m:m<\omega\rangle$ is well-defined by Claims~\ref{c:gridd}(iii), \ref{c:int} and \eqref{gridunique}$^\bullet$.
We show by induction on $m$ that for all $m<\omega$, 
$\bigl\auf \langle q_0,{\bf c}(0)\rangle,\dots,\langle q_m,{\bf c}(m)\rangle\bigr\zu$
is a run of $M$ starting with all-0 counters. Indeed, $c_i(0)=0$ for $i<N$ by $\coufd$ 
and Claim~\ref{c:seesdd}.
Now suppose the statement holds for some $m<\omega$. By Claim~\ref{c:gridd}, 
$\M,\auf \Ix(x_{m+1}),y_{m+1}\zu\models\state\land \Dv( \diag\land \DhM\state_{q_m})$. So by \eqref{executebwdec}$^\bullet$, there is 
$\langle\iota,q_{m+1}\rangle\in I_{q_m}$ such that 
$\M,\auf x_{m+1},y_{m+1}\zu\models\insbw\land\state_{q_{m+1}}$, and so 
$\M,\auf \Ix(x_{m+1}),y_{m+1}\zu\models\insbw$ by $\iotaint$ and Claim~\ref{c:int}.
Thus, $\M,\auf \Ix(x_{m+1}),y_{m+1}\zu\models\exeibwd$ by \eqref{instbw}$^\bullet$. It follows from Claim~\ref{c:counterd} that 
$\sigma_m\stepi\sigma_{m+1}$ as required.
\end{proof}

For the other direction, suppose that $M$ has an infinite run starting with all-0 counters.
Let $\F=\auf W,R\zu$ be a \preorder in $\CC$ containing an $\auf\omega + 1, >\zu$-type chain
$x_\omega R\dots R x_m R \dots R x_0$.
For every $m<\omega$, we let
\[
[x_{m+1},x_m):=\bigl(\{w\in W : x_{m+1} R w R x_m\}\cup\{x_{m+1}\}\bigr)-
\{w : w=x_{m}\mbox{ or }x_m R w\}.
\]
Take the model $\M^\infty=\bigl\auf\auf\omega + 1, >\zu\mprod\auf\omega,\ne\zu,\mu\bigr\zu$  defined in  \eqref{minftyfirst}--\eqref{minftylast}.
We define a model $\N^\infty=\bigl\auf\F\mprod\auf\omega,\ne\zu,\nu\bigr\zu$ as follows.
We let
\[
\nu( \tick):=\{\auf w,n\zu : w\in[x_{m+1},x_m),\ m,n<\omega,\ m\mbox{ is odd}\},
\]
for all  $\pred\in\{\diag,\state,\state_q,\insbw,\cou\}_{q\in Q,\,\iota\in\textit{Op}_C,\,i<N}$,
\[
\nu( \pred) :=\{\auf w,n\zu :  w\in[x_{m+1},x_m),\ \auf m,n\zu\in\mu(\pred)\ \mbox{for some }m<\omega\},
\]
for all  $\pred\in\{\diag,\state,\state_q,\insbw\}_{q\in Q,\,\iota\in\textit{Op}_C}$,
\[
 \nu(\predp):=\{\auf w,n\zu :  w\in[x_{m},x_{m-1}),\ \auf m,n\zu\in\mu(\pred)\ \mbox{for some }m>0\},
\]
and for all $i<N$,
\begin{align*}
\nu(\minusi) & := \{\auf w,n\zu :  w\in[x_{m+1},x_{m}),\ \auf m,n\zu\notin\mu(\cou),\
\ \auf m-1,n\zu\in\mu(\cou)\ \mbox{for some }m>0\},\\
\nu(\minusip) & := \{\auf w,n\zu :  w\in[x_{m+1},x_{m}),\ \auf m,n\zu\in\mu(\cou),\
\ \auf m+1,n\zu\notin\mu(\cou)\ \mbox{for some }m<\omega\}.
\end{align*}
It is not hard to check that $\N^\infty,\auf x_\omega,0\zu\models\diagbwd\land\fmmbwdecd$.
So by Lemma~\ref{l:runbwd}, CM non-termination is reducible to $\CC\dprod\cdiff$-satisfiability.
This completes the proof of Theorem~\ref{t:main}.


\section{Expanding 2-frames}\label{domain}

In this section we show that satisfiability over classes of expanding 2-frames can be genuinely simpler
than satisfiability over the corresponding product frame classes, but it is still quite complex.


 \subsection{Lower bounds}\label{expprodl}


  \begin{theorem}\label{t:omegael}
  $\{\auf\omega,<\zu\}\eprod\cdiff$-satisfiability is undecidable.
 \end{theorem}
 
\begin{corollary}\label{co:omegaefol}
\logic-satisfiability is undecidable in expanding domain models over $\auf\omega,<\zu$.
\end{corollary}

\begin{theorem}\label{t:decefin}
$\clinf\eprod\cdiff$-satisfiability is Ackermann-hard.
\end{theorem}

\begin{corollary}\label{co:decefinfol}
\logic-satisfiability is Ackermann-hard in expanding domain models over the class of all 
finite linear orders.
\end{corollary}

We prove Theorem~\ref{t:omegael}
by reducing the `LCM $\omega$-reachability' problem to \mbox{$\{\auf\omega,<\zu\}\eprod\cdiff$}-satisfiability.
The  idea of our reduction is similar to the one used in \cite{KonevWZ05} for a 
more expressive formalism.  It is sketched in Fig.~\ref{f:model}: 
First, we generate an infinite diagonal staircase going forward. Then, still going forward, we place longer and longer finite runs one after the other. However, each individual run proceeds backward.
Also, we can force only lossy runs this way. When going backward horizontally in expanding 2-frames, the vertical columns might become
smaller and smaller, so some of the points carrying the information on the content of the counters might 
disappear as the runs  progress.
\begin{figure}[ht]
\begin{center}
\setlength{\unitlength}{.08cm}
\begin{picture}(130,72)(-5,0)
\thicklines
\put(5,10){\line(0,1){25}}
\put(25,10){\line(0,1){32}}
\put(55,10){\line(0,1){45}}
\put(95,10){\line(0,1){52}}

\thinlines
\put(5,10){\vector(1,0){100}}
\put(5,8){$\underbrace{\hspace*{1.55cm}}_{\leftarrow\ \rho_1}$}
\put(0,70){${}_{\start}$}
\put(4,65){${}_{\downarrow}$}
\put(5,10){\line(2,1){95}}
\multiput(15,15)(10,5){9}{\circle*{1.5}}
\put(25,8){$\underbrace{\hspace*{2.35cm}}_{\leftarrow\ \rho_2}$}
\put(20,70){${}_{\start}$}
\put(24,65){${}_{\downarrow}$}
\put(55,8){$\underbrace{\hspace*{3.15cm}}_{\leftarrow\ \rho_3}$}
\put(50,70){${}_{\start}$}
\put(54,65){${}_{\downarrow}$}
\put(90,70){${}_{\start}$}
\put(94,65){${}_{\downarrow}$}
\put(94,63){${}_{}$}
\put(108,9){$\dots$}
\put(103,59){$\dots$}
\put(98,3){$\dots$}
\put(117,9){$\langle\omega,<\rangle$}

\put(5,15){\line(1,0){10}}
\put(5,15){\circle*{1.5}}
\put(0,16){${}_{\recx}$}
\put(15,25){\line(1,0){20}}
\put(15,25){\line(0,-1){10}}
\put(15,25){\circle*{1.5}}
\put(10,26){${}_{\recx}$}
\put(25,30){\line(1,0){20}}
\put(25,30){\circle*{1.5}}
\put(20,31){${}_{\recx}$}
\put(35,40){\line(1,0){30}}
\put(35,40){\line(0,-1){15}}
\put(35,40){\circle*{1.5}}
\put(30,41){${}_{\recx}$}
\put(45,45){\line(1,0){30}}
\put(45,45){\line(0,-1){15}}
\put(45,45){\circle*{1.5}}
\put(40,46){${}_{\recx}$}
\put(55,50){\line(1,0){30}}
\put(55,50){\circle*{1.5}}
\put(50,51){${}_{\recx}$}

\put(16,13){${}_{\state_{q_r}}$}
\put(26,17){${}_{\state_{q_0}}$}
\put(36,22){${}_{\state_{q_r}}$}
\put(46,27){${}_{\state_{q_r}}$}
\put(56,32){${}_{\state_{q_0}}$}
\put(66,37){${}_{\state_{q_r}}$}
\put(76,42){${}_{\state_{q_r}}$}
\put(86,47){${}_{\state_{q_r}}$}
\put(96,52){${}_{\state_{q_0}}$}
\end{picture}
\end{center}
\caption{Representing longer and longer $n$-recurrent lossy runs $\rho_n$ in 2-frames expanding over $\auf\omega,<\zu$.}\label{f:model}
\end{figure}
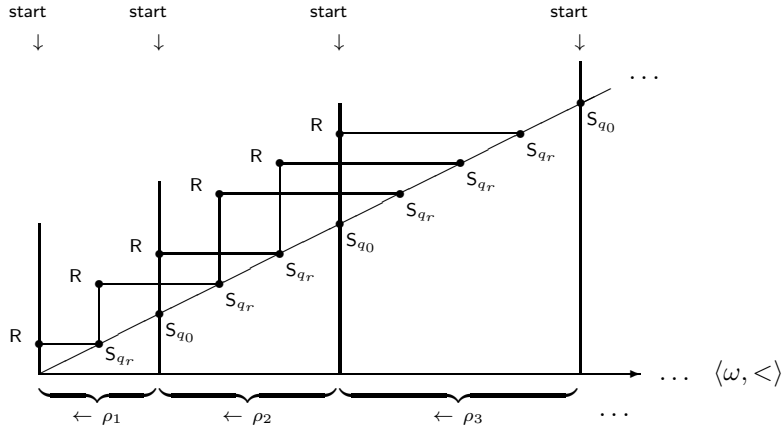

To this end,
let $\Hh_{\auf\omega,<\zu,\overline{\G}}$ be an expanding 2-frame for some difference frames
$\G_n=\auf W_n,\ne\zu$, $n<\omega$, and let $\M$ be a model based on $\Hh_{\auf\omega,<\zu,\overline{\G}}$.
First, we generate an infinite diagonal staircase forward in $\M$, similarly how we did in the 
proof of Theorem~\ref{t:omega}. However, this time we use the vertical counting capabilities to
force the \emph{uniqueness} of this staircase. To this end, let $\diagfwu$ be the conjunction of \eqref{initfw}--\eqref{sgenfw} and
\begin{equation}\label{diaguniq}
 \Bh^+\Bv\bigl(\diag\to\Bv\neg\diag).
\end{equation}

The following `expanding generalisation' of Claim~\ref{c:gridfw} can be proved by a straightforward induction on $m$:

\begin{claim}\label{c:gridfwe}
Suppose that $\M,\auf 0,r\zu\models\diagfwu$. Then there exists a sequence 
\mbox{$\auf y_m : m<\omega\zu$} such that for all $m<\omega$,
\begin{enumerate}
\item[{\rm (}i{\rm )}]
$y_0=r$ and if $m>0$ then $y_{m}\in W_{m-1}$,
\item[{\rm (}ii{\rm )}]
for all $n<m$, $y_m\ne y_{n}$, 
\item[{\rm (}iii{\rm )}]
$\M,\auf m,y_m\zu\models\state$,
\item[{\rm (}iv{\rm )}]
for all $w\in W_m$, $\M,\auf m,w\zu\models\diag$ iff $w=y_{m+1}$.
\end{enumerate}
\end{claim}

Given a counter machine $M$,
we will encode lossy runs that start with all-0 counters by going backward along the created 
diagonal staircase.
We will adjust the tools developed in the proof of Theorem~\ref{t:mainirr} in order to 
handle lossyness, and also to force not just one run, but
several (finite) runs, placed one after the other.
To this end, we introduce a fresh propositional variable $\start$, intended to mark the start
of each run (see Fig.~\ref{f:model}),  and for each $i<N$ we let
\[
\allcl:=\ \ \Dh\diag\land\Bh\bigl(\diag\lor\Dh\diag\to(\neg\start\land\cou)\bigr).
\]
Then we have the following lossy analogue of Claims~\ref{c:counting} and \ref{c:countingd}:
\begin{claim}\label{c:countingl}
Suppose that $\M,\auf 0,r\zu\models\diagfwu$. Then for all $m<\omega$, $i<N$,
\[
\{w\in W_m : \M,\auf m,w\zu\models\allcl\}\subseteq \{w\in W_{m+1} :\M,\auf m+1,w\zu\models\cou\}.
\]
\end{claim}

\begin{proof}
Suppose that $ \M,\auf m,w\zu\models\allcl$. Then  $\M,\auf m,w\zu\models\Dh\diag$ and
so by Claim~\ref{c:gridfwe}(iv), $w=y_n$ for some $n>m+1$, and we have
$\M, \auf n-1,w\zu\models\diag$. Thus, $\M, \auf m+1,w\zu\models\diag\lor\Dh\diag$.
As $\M,\auf m,w\zu\models \Bh(\diag\lor\Dh\diag\to\cou)$, we obtain 
$\M, \auf m+1,w\zu\models\cou$ as required.
\end{proof}

Now, for each $i<N$, we can simulate the possible lossy changes in the counters by the following formulas:
\begin{align*}
\ffixbwl & :=\ \ \Bv^+ (\cou\to\allcl),\\
\fincbwl & :=\ \ \Bv^+\bigl(\cou\to (\diag\lor\allcl)\bigr),\\
\fdecbwl & :=\ \ \Bv^+(\cou\to\allcl)\land\Dv^+(\neg\cou\land\allcl).
\end{align*}
The following lossy analogue of Claims~\ref{c:counter} and \ref{c:counterd} is a straightforward consequence of Claims~\ref{c:gridfwe}(iv) and \ref{c:countingl}.
Note that the vertical uniqueness of $\diag$-points is used in simulating the lossy incrementation steps properly. 

\begin{claim}\label{c:counterl}
Suppose that $\M,\auf 0,r\zu\models\diagfwu$. For all $i<N$, $m<\omega$,
let $c_i(m) :=|\{w\in W_m: \M,\auf m,w\zu\models\cou\}|$.
Then for all $m<\omega$,
\[
c_i(m)\leq \left\{
\begin{array}{ll}
c_i(m+1), & \mbox{ if $\M,\auf m,y_{m}\zu\models\ffixbwl$},\\[3pt]
c_i(m+1)+1,\ & \mbox{ if $\M,\auf m, y_{m}\zu\models\fincbwl$},\\[3pt]
c_i(m+1)-1, & \mbox{ if $\M,\auf m,y_{m}\zu\models\fdecbwl$}.
\end{array}
\right.
\]
\end{claim}

Next, we encode the various counter machine instructions for lossy steps, acting backward. For each $\iota\in\textit{Op}_C$, we define the formula $\exeil$ by taking
\[
\exeil:=\ \ \left\{
\begin{array}{ll}
\displaystyle\fincbwl\land\bigwedge_{i\ne j<N}\ffixjbwl, & \mbox{ if $\iota=\finci$},\\
\displaystyle\fdecbwl\land\bigwedge_{i\ne j<N}\ffixjbwl, & \mbox{ if $\iota=\fdeci$},\\
\displaystyle\Bv^+\neg\cou\land\bigwedge_{j<N}\ffixjbwl, & \mbox{ if $\iota=\ftest$}.\\
\end{array}
\right.
\]
Finally, given a counter machine $M$, we encode lossy runs that start with all-0 counters at
$\start$-marks, and go backward until the next $\start$-mark.
We define $\fmmbwlossy$ to be the conjunction of  the following formulas:
\begin{align}
\label{startv}
& \Bh^+\Bv^+\bigl(\start\to\Bv\start\bigr),\\
 \label{griduniquel}
& \Bh^+\Bv^+\bigl(\state\leftrightarrow\bigvee_{q\in Q-H} \bigl(\state_q\land\bigwedge_{q\ne q'\in Q}\neg\state_{q'})\bigr),\\
\label{initmmbwl}
& \Bh^+\Bv^+\bigl(\state\land\start\to (\state_{q_0}\land\bigwedge_{i<N}\Bv^+\neg \cou)\bigr),\\
\label{executebwl}
& \Bh\Bv\bigwedge_{q\in Q-H}
\bigl[\bigl(\state\land\neg\start\land \Dv( \diag\land \Dh\state_q)\bigr) \to
\bigvee_{\langle\iota,q'\rangle\in I_q}(\exeil\land\state_{q'})\bigr].
\end{align}
Then we have the following lossy analogue of Lemmas~\ref{l:runbw} and \ref{l:runbwd}:

\begin{claim}\label{c:runbwl}
Suppose $\M,\auf 0,r\zu\models\diagfwu\land\fmmbwlossy$, and
for all $m<\omega$, $i<N$ , let 
\[
 s_m  :=q, \mbox{ if }\M,\auf m,y_m\zu\models\state_{q},\ \
 c_i(m)  :=|\{w\in W_m: \M,\auf m,w\zu\models\cou\}|,\ \
 \sigma_m  :=\auf s_m,{\bf c}(m)\zu.
 \]
Then $\auf \sigma_{a},\sigma_{a-1},\dots,\sigma_{b}\zu$
is a well-defined lossy run of $M$ starting with  $\auf q_0,{\bf 0}\zu$,
whenever $b< a<\omega$ is such that $\M,\auf a,r\zu\models\start$,
and $\M,\auf n,r\zu\models\neg\start$, for every $n$ with $b\leq n<a$.
\end{claim}

\begin{proof}
The sequence $\auf s_{a},s_{a-1},\dots,s_{b}\zu$ 
is well-defined by Claim~\ref{c:gridfwe}(iii) and \eqref{griduniquel}.
We show by induction on $m$ that for all $m\leq a-b$, 
$\auf \sigma_{a},\sigma_{a-1},\dots,\sigma_{a-m}\zu$
is a lossy run of $M$ starting with $\auf q_0,{\bf 0}\zu$.
Indeed, $\M,\auf a,y_{a}\zu\models\start$ by \eqref{startv}, and so 
$s_{a}=q_0$ and $c_i(a)=0$ for $i<N$ by Claim~\ref{c:gridfwe}(iii) and
\eqref{initmmbwl}.
Now suppose the statement holds for some $m< a-b$.
As $\M,\auf a-m-1,y_{a-m-1}\zu\models\neg\start$ by \eqref{startv}, we have 
\[
\M,\auf a-m-1,y_{a-m-1}\zu\models\state\land\neg\start\land \Dv( \diag\land \Dh\state_{s_{a-m}})
\]
by Claim~\ref{c:gridfwe}.
By \eqref{griduniquel} we have $s_{a-m}\in Q-H$, and so by \eqref{executebwl} there is 
$\langle\iota,s_{a-m-1}\rangle\in I_{s_{a-m}}$ such that 
$\M,\auf a-m-1,y_{a-m-1}\zu\models\exeil$.
It follows from Claims~\ref{c:countingl} and \ref{c:counterl} that 
$\sigma_{a-m}\lstepi\sigma_{a-m-1}$ as required.
\end{proof}

It remains to force that the $n$th run visits $q_r$ at least $n$ times. 
To this end, we introduce two fresh propositional variables $\recx$ and $\recp$, and
define $\frece$ as the conjunction of \eqref{startv} and the following formulas:
\begin{align}
\label{recinit}
& \start\land\Bh^+\Dh\start,\\
\label{startpoints}
& \Bh^+\Bv^+\bigl[\start\to\Dv^+\bigl(\recx\land\Dh(\state\land\neg\start)\land\Bh(\Dh\state\to\neg\start)\bigr)\bigr],\\
\label{qr}
& \Bh^+\Bv^+\big(\recx\to\Bh(\state\to\recp)\bigr),\\
\nonumber
& \Bh\Bv^+\Bigl[\recp\to\Dv\bigl[\recx\land\Dh\bigl(\start\land\Dh(\state\land\neg\start)\bigr)\,\land\\[-3pt]
\label{recpoints}
&\hspace*{3.5cm}
\Bh\bigl(\start\land\Dh\state\to\Bh(\Dh\state\to\neg\start)\bigr)\bigr]\Bigr],\\
\label{sstars}
& \Bh^+\Bv^+(\recp\to\state),\\
\label{svuniq}
& \Bh^+\Bv^+(\state\to\Bv\neg\state),\\
\label{unipoints}
& \Bh^+\Bv^+(\recx\to\Bh\neg\recx).
\end{align}
\begin{claim}\label{c:frecfw}
Suppose that $\M,\auf 0,r\zu\models\diagfwu\land\frece$. 
Then there is an infinite sequence
$\langle k_n : n<\omega\rangle$ such that, for all $n<\omega$,
\begin{itemize}
\item[{\rm (}i{\rm )}]
$\M,\auf k_n,w\zu\models\start$ for all $w\in W_{k_n}$, 
\item[{\rm (}ii{\rm )}]
if $n>0$ then $\M,\auf k,w\zu\models\neg\start$ for all $k$ with $k_{n-1}<k<k_n$ and $w\in W_k$, and
\item[{\rm (}iii{\rm )}]
if $n>0$ then $|\{k : k_{n-1}<k< k_n\mbox{ and }\M,\auf k,y_k\zu\models\recp\}|\geq n$.
\end{itemize}
\end{claim}

\begin{proof}
By induction on $n$. To begin with, let $k_0=0$.
Now suppose inductively that we have $\langle k_\ell : \ell<n\rangle$ 
as required, for some $0<n<\omega$.
Now let $k_n$ be the smallest $k$ with $k>k_{n-1}$ and $\M,\auf k,r\zu\models\start$
(there is such by \eqref{recinit}). So $k_n>k_{n-1}$, and by \eqref{startv} 
\begin{equation}\label{nextstart}
\M,\auf k_n,w\zu\models\start\mbox{ for all $w\in W_{k_n}$.}
\end{equation}
As by the IH(i) we have $\M,\auf k_{n-1},r\zu\models\start$, by \eqref{startpoints} there is $w\in W_{k_{n-1}}$
such that
\[
\M,\auf k_{n-1},w\zu\models 
\recx\land\Dh(\state\land\neg\start)\land\Bh(\Dh\state\to\neg\start).
\]
By Claim~\ref{c:gridfwe}(iii) and \eqref{svuniq}, $w=y_{i_n}$ for some $k_{n-1}<i_n< k_n$,
and so $\M,\auf i_n,y_{i_n}\zu\models\recp$ follows by \eqref{qr}.
In particular, if $n=1$ then $\M,\auf i_1,y_{i_1}\zu\models\recp$, and so 
\[
|\{k : k_{0}<k< k_1\mbox{ and }\M,\auf k,y_k\zu\models\recp\}|\geq 1.
\]
Now suppose that  $n>1$ and take some $k$ such that $k_{n-2}<k< k_{n-1}$ and $\M,\auf k,y_k\zu\models\recp$.
By \eqref{recpoints}, there is $v\in W_k$ such that
\begin{equation}\label{rpoint}
\M,\auf k,v\zu\models\recx\land\Dh\bigl(\start\land\Dh(\state\land\neg\start)\bigr)
\land\Bh\bigl(\start\land\Dh\state\to\Bh(\Dh\state\to\neg\start)\bigr).
\end{equation}
So there is some $k'>k$ with
$\M,\auf k',v\zu\models\start\land\Dh(\state\land\neg\start)$, and so by the IH we have 
\begin{equation}\label{zzzz}
\M,\auf k_{n-1},v\zu\models\start\land\Dh(\state\land\neg\start).
\end{equation}
Therefore, by \eqref{rpoint} we have
\begin{equation}\label{wwww}
 \M,\auf k_{n-1},v\zu\models\Bh(\Dh\state\to\neg\start).
 \end{equation}
 By \eqref{zzzz}, there is some $k^+>k_{n-1}$ with $\M,\auf k^+,v\zu\models\state\land\neg\start$.
Therefore, $v=y_{k^+}$ by Claim~\ref{c:gridfwe}(iii) and \eqref{svuniq}, 
$\M,\auf k^+,y_{k^+}\zu\models\recp$ by \eqref{qr}, and $k^+\ne k_n$ by \eqref{nextstart}.
Moreover, we have that $k^+< k_n$ because of the following. If $k^+>k_n$ were the case,
then $\M,\auf k_n,v\zu\models\Dh\state$, and so $\M,\auf k_n,v\zu\models\neg\start$ by
\eqref{wwww}, contradicting \eqref{nextstart}.
Further, by \eqref{unipoints} we obtain that $k^+\ne i_n$, and $k^+\ne\ell^+$ 
whenever $k\ne\ell$, $k_{n-1}< k,\ell<k_n$.
Therefore, by \eqref{sstars}, \eqref{svuniq}, and the IH(iii), we have
\begin{multline*}
|\{k : k_{n-1}<k< k_n\mbox{ and }\M,\auf k,y_k\zu\models\recp\}|\geq\\
 |\{k : k_{n-2}<k< k_{n-1}\mbox{ and }\M,\auf k,y_k\zu\models\recp\}|+1\geq n-1+1=n,
\end{multline*}
as required.
\end{proof}

Now the following lemma is a straightforward consequence of Claims~\ref{c:runbwl} and \ref{c:frecfw}:

\begin{lemma}\label{l:lrecpred}
Suppose $\M,\auf 0,r\zu\models\diagfwu\land\fmmbwlossy\land\frece\land\Bh^+\Bv^+(\recp\to\state_{q_r})$.
Then, for every $n<\omega$, $M$ has a lossy run starting with $\auf q_0,{\bf 0}\zu$ 
and visiting $q_r$ at least $n$ times.
\end{lemma}

On the other hand, suppose that for every $0<n<\omega$, $M$ has a lossy run 
\[
\rho_n=\bigl \auf\auf q_0^n,{\bf 0}\zu,\dots,\auf q_{m_{n}-1}^n,{\bf c}(m_n-1)\zu\bigr\zu
\]
such that $q_0^n=q_0$ and $\rho_n$ visits $q_r$
at least $n$ times. Let $M_0:=0$ and for each $0<n<\omega$, let $M_n:=\sum_{i=1}^n m_i$, and let
$i_1^n,\dots,i_n^n< m_n$ be such that
$|\{i_1^n,\dots,i_n^n\}|=n$ and $q_i=q_r$ for every $i\in\{i_1^n,\dots,i_n^n\}$.
We define a model $\Ninf=\bigl\auf\auf\omega,<\zu\mprod\auf\omega,\ne\zu,\alpha\bigr\zu$ as follows (cf.\  Fig.~\ref{f:model}):
For all $q\in Q$, we let
\begin{align*}
\alpha(\state_q)& :=\{\auf n,n\zu : M_k\leq n< M_{k+1}\mbox{ and }q_{n-M_k}^{k+1}=q,\mbox{ for some $k<\omega$}\},\\
\alpha(\state)& :=\{\auf n,n\zu : n<\omega\},\\
\alpha(\diag)& :=\{\auf n,n+1\zu : n<\omega\},\\
\alpha(\start) &:=\{\auf n,m\zu: n=M_k\mbox{ for some $k<\omega$, and }m<\omega\}.
\end{align*}
Further, for any finite subset $X=\{n_1,\dots,n_\ell\}$ of $\omega$ with $n_1<\dots < n_\ell$ and any $k\leq |X|$, we let
$\textit{min}_k(X):=\{n_1,\dots,n_k\}$. Now
for all $i<N$, $0<n<\omega$ and $k<m_n$, we define the sets $\alpha_k^n(\cou)$ by
induction on $k$:
We let $\alpha_0^n(\cou):=\emptyset$, and for all $k<m_n-1$,
\[
\alpha_{k+1}^n(\cou):=\left\{
\begin{array}{ll}
\alpha_k^n(\cou)\cup\{M_{n}-k\}, & \mbox{ if $c_i^n(k+1)=c_i^n(k)+1$},\\
\alpha_k^n(\cou)-\textit{min}_\ell\bigl(\alpha_k^n(\cou)\bigr), & \mbox{ if $|c_i^n(k)-c_i^n(k+1)|=\ell$}.
\end{array}
\right.
\]
Then, for each $i<N$, we let
\[
\alpha(\cou):=\{\auf k,m\zu : M_{n-1}\leq k<M_n,\ m\in\alpha_{M_{n}-k-1}^n(\cou)\mbox{ for some $0<n<\omega$}\}.
\]
Also, we define the sequence $\auf r_n:n<\omega\zu$ inductively as follows. Let $r_0:=i_1^1$ and
let
\[
r_{n+1}:=\left\{
\begin{array}{ll}
M_k+i_1^{k+1}, & \mbox{if $r_n=M_{k-1}+i_k^k$ for some $k>0$},\\[3pt]
M_{k-1}+i_{\ell+1}^k, & \mbox{if $r_n=M_{k-1}+i_\ell^k$ for some $k>0,\ \ell<k$}.
\end{array}
\right.
\]
Then let
\begin{align*}
\alpha(\recp) &:=\{\auf r_n,r_n\zu : n<\omega\},\\
\alpha(\recx) &:=\{\auf n,r_n\zu : n<\omega\}.
\end{align*}
It is not hard to check that 
$\Ninf,\auf 0,0\zu\models\diagfwu\land\fmmbwlossy\land\frece\land\Bh^+\Bv^+(\recp\to\state_{q_r})$, and so by Lemma~\ref{l:lrecpred}
LCM \mbox{$\omega$-reachability} can be reduced to
$\{\auf\omega,<\zu\}\eprod\cdiff$-satisfiability. This competes the proof of Theorem~\ref{t:omegael}.
 
 
 \bigskip
Next, we prove Theorem~\ref{t:decefin}
by reducing the `LCM-reachability' problem to $\clinf\eprod\cdiff$-satisfiability.
We will use the finitary versions of some of the formulas used in the previous proof.
Let $\Hh_{\auf T,<\zu,\overline{\G}}$ be an expanding 2-frame for some finite linear order $\auf T,<\zu$ and for some difference frames $\G_n=\auf W_n,\ne\zu$, $n\in T$, and let $\M$ be a model based on $\Hh_{\auf T,<\zu,\overline{\G}}$.
We may assume that $T=|T|<\omega$.
We consider a version of the formula $\diagfwfin$ defined in the proof of Theorem~\ref{t:finite}.
Let $\diagfwfinu$ be the conjunction of \eqref{initfw}, \eqref{dgenfw}, \eqref{sgenfin} and \eqref{diaguniq}.
The following finitary version of Claim~\ref{c:gridfwe} can be proved by a straightforward induction on $m$:

\begin{claim}\label{c:gridfwfinexp}
Suppose $\M,\auf 0,r\zu\models\diagfwfinu$. Then there exist some $0<E\leq T$ and a sequence
$\auf y_m : m\leq E\zu$ of points such that for all $m\leq E$,
\begin{enumerate}
\item[{\rm (}i{\rm )}]
$y_0=r$ and if $m>0$ then $y_{m}\in W_{m-1}$,
\item[{\rm (}ii{\rm )}]
for all $n<m$, $y_m\ne y_{n}$, 
\item[{\rm (}iii{\rm )}]
if $m< E$ then $\M,\auf m,y_m\zu\models\state$,
\item[{\rm (}iv{\rm )}]
if $m<E$ then for all $w\in W_m$, $\M,\auf m,w\zu\models\diag$ iff  $w=y_{m+1}$,
\item[{\rm (}v{\rm )}]
$\M,\auf E-1,y_E\zu\models\pend$, and if $m<E-1$ then
$\M,\auf m,y_{m+1}\zu\models\neg\pend$.
\end{enumerate}
\end{claim}

The following lemma is a straightforward consequence of Claims~\ref{c:runbwl} and \ref{c:gridfwfinexp}:

\begin{lemma}\label{l:runfwdfinexp}
Suppose that $\M,\auf 0,r\zu\models\diagfwfinu\land\fmmbwlossy\land
\state_{q_r}\land\Bh^+\Bv^+(\pend\leftrightarrow\start)$. 
For all  $m< E$ and $i<N$, let 
\[
s_m:=q,\ \mbox{ if }\ \M,\auf m,y_m\zu\models\state_q,\ \
c_i(m):=|\{w\in W : \M,\auf m,w\zu\models\cou\}|,\ \
\sigma_m=\auf s_m,{\bf c}(m)\zu.
\]
Then $\auf \sigma_{E-1},\sigma_{E-2},\dots,\sigma_0\zu$ is a well-defined lossy run
of $M$ starting with $\auf q_0,{\bf 0}\zu$ and reaching $q_r$.
\end{lemma}

On the other hand,  if $M$ has a run $\bigl\auf \auf q_m,{\bf c}(m)\zu: m<T\bigr\zu$ for some $T<\omega$ such that it starts with all-0 counters and $q_{T-1}=q_r$, then it is not hard to define a model 
based on $\auf T,<\zu\mprod\auf T+1,\ne\zu$ satisfying 
$\diagfwfinu\land\fmmbwlossy\land
\state_{q_r}\land\Bh^+\Bv^+(\pend\leftrightarrow\start)$ (cf.\ how the finite runs in the model $\Ninf$ are defined in the proof of 
Theorem~\ref{t:omegael}).
So by Lemma~\ref{l:runfwdfinexp} the proof of Theorem~\ref{t:decefin} is completed.


 \subsection{Upper bounds}\label{expprodu}
 
 To begin with, as a consequence of Theorems~\ref{t:re} and Props.~\ref{p:cd}, \ref{p:foprod} we obtain:

\begin{corollary}\label{co:linefou}
\logic-satisfiability is co-r.e.\ in expanding domain models over the class of all linear orders.
\end{corollary}

Unlike in the constant domain case, in the expanding domain case
the same holds for $\auf\omega,<\zu$ as timeline:

   \begin{theorem}\label{t:omegaeu}
  $\{\auf\omega,<\zu\}\eprod\cdiff$-satisfiability is co-r.e.
 \end{theorem}

\begin{corollary}\label{co:omegaefou}
\logic-satisfiability is co-r.e.\ in expanding domain models over $\auf\omega,<\zu$.
\end{corollary}

\begin{theorem}\label{t:decefinu}
$\clinf\eprod\cdiff$-satisfiability is decidable.
\end{theorem}
 
 \begin{corollary}\label{co:decefinfou}
\logic-satisfiability is decidable in expanding domain models over the class of all 
finite linear orders.
\end{corollary}
 
 In order to prove both Theorems~\ref{t:omegaeu} and \ref{t:decefinu},
we begin with showing
 that there is a reduction from $\cdiff$-satisfiability to $\clin$-satisfiability that can be `lifted to the 2D level'.
As we will use this reduction to obtain upper bounds on satisfiability in expanding 2-frames,
we formulate it in this setting only. 
To this end, fix some bimodal formula $\phi$.
For every $\psi\in\subf$, we introduce a fresh
propositional variable $\psivar$ not occurring in $\phi$, and define inductively a translation $\psi^\dagger$ by taking
\begin{align*}
\pred^\dagger & := \ \pred, \mbox{ for each propositional variable $\pred\in\subf$},\\
(\neg\psi)^\dagger & :=\  \neg\psi^\dagger,\\
(\psi_1\land\psi_2)^\dagger & :=\ \psi_1^\dagger\land\psi_2^\dagger,\\
(\Dh\psi)^\dagger & := \ \Dh\psi^\dagger,\\
(\Dv\psi)^\dagger & := \ \psivar\lor\Dv\psi^\dagger.
\end{align*}
Further, we let 
\[
\expform:=\ \
\Bh^+\!\!\bigwedge_{\psi\in\subf}\!\!
\neg\psivar\land
\Bv^+(\psi^\dagger\to\Bv\psivar)\land
\bigl(\Dv\psivar\to\Dv^+(\neg\psivar\land\psi^\dagger)\bigr).
\]
\begin{claim}\label{c:difftolin}
For any formula $\phi$, and any class $\CC$ of transitive frames,
\begin{itemize}
\item
$\phi$ is $\CC\eprod\cdiff$-satisfiable iff 
$\expform\land\phi^\dagger$ is $\CC\eprod\clin$-satisfiable.
\item
$\phi$ is $\CC\eprod\cdifff$-satisfiable iff 
$\expform\land\phi^\dagger$ is $\CC\eprod\clinf$-satisfiable.
\end{itemize}
\end{claim}
\begin{proof}
$\Rightarrow$:
Suppose that
$\M,\auf r_0,r_1\zu\models\phi$ in some model $\M=\auf \Hh_{\F,\overline{G}},\mu\zu$ based on an expanding 2-frame
$\Hh_{\F,\overline{\G}}$ where $\F=\auf W,R\zu$ is transitive and for every $x\in W$, $\G_x=\auf W_x,\ne\zu$.
Then $W_x\subseteq W_y$ whenever $xRy$, $x,y\in W$. Also,
we may assume that $r_0$ is a root in $\F$, and so $r_1\in W_x$ for all $x\in W$.
So for every $x\in W$ we may take a well-order $<_x$ on $W_x$ with least
element $r_1$ and such that $<_x\subseteq <_y$ whenever $xRy$. Let $\Sigma_x'=\auf W_x,<_x\zu$, for $x\in W$.
Then clearly $\Hh_{\F,\overline{G}'}\in \CC\eprod\clin$. We define a model $\M'=\auf \Hh_{\F,\overline{G}'},\mu'\zu$ 
by taking
\begin{align*}
\mu'(\pred)& :=\mu(\pred), \mbox{ for $\pred\in\subf$},\\
\mu'(\psivar) & := \{\auf x,w\zu : x\in W\mbox{ and } \M,\auf x,u\zu\models\psi\mbox{ for some $u\in W_x$ with $u<_x w$}\}.
\end{align*}

First, we show by induction on $\psi$ that for all $\psi\in\subf$, $x\in W$, $u\in W_x$,
\begin{equation}\label{expih}
\M,\auf x,u\zu\models\psi\qquad\mbox{iff}\qquad\M',\auf x,u\zu\models\psi^\dagger.
\end{equation}
Indeed, the only non-straightforward case is that of $\Dv$. So suppose first that $\M,\auf x,u\zu\models\Dv\psi$.
Then there is $v\in W_x$, $v\ne u$ with $\M,\auf x,v\zu\models\psi$.
If $u<_x v$ then $\M',\auf x,u\zu\models\Dv\psi^\dagger$ by the IH. If $v<_x u$, then 
$\M',\auf x,u\zu\models\psivar$ by the definition of $\M'$. So in both cases we have $\M',\auf x,u\zu\models(\Dv\psi)^\dagger$.
Conversely, suppose that $\M',\auf x,u\zu\models(\Dv\psi)^\dagger$.
If $\M',\auf x,u\zu\models\psivar$ then there is $v\in W_x$, $v<_x u$ with $\M,\auf x,v\zu\models\psi$.
Therefore, there is $v\in W_x$, $v\ne u$ with $\M,\auf x,v\zu\models\psi$.
If $\M',\auf x,u\zu\models\Dv\psi^\dagger$ then there is $v\in W_x$, $v<_x u$ with $\M',\auf x,v\zu\models\psi^\dagger$,
and so there is $v\in W_x$, $v\ne u$ with $\M,\auf x,v\zu\models\psi$ by the IH. 
So in both cases $\M,\auf x,u\zu\models\Dv\psi$ follows.

Second, we claim that $\M',\auf r_0,r_1\zu\models\expform$. 
Indeed, take any $x\in W$. As $r_1$ is $<_x$-least in $W_x$, we have $\M',\auf x,r_1\zu\models\neg\psivar$.
Now take any $y\in W_x$ with $\M',\auf x,y\zu\models\psi^\dagger$ and suppose that $y<_x z$ for some $z\in W_x$.
By \eqref{expih}, we have $\M,\auf x,y\zu\models\psi$ and so $\M',\auf x,z\zu\models\psivar$ by the definition of $\M'$.
Finally, suppose that $\M',\auf x,r_1\zu\models\Dv\psivar$. Therefore, the set 
$\{w\in W_x : \auf x,w\zu\in\mu'(\psivar)\}$ is non-empty. Let $y$ be its $<_x$-least element.
So there is $z\in W_x$, $z<_x y$ such that $\M,\auf x,z\zu\models\psi$ and $\auf x,z\zu\notin\mu'(\psivar)$.
Thus $\M',\auf x,z\zu\models\neg\psivar\land\psi^\dagger$ by \eqref{expih}.
As either $r_1=y$ or $r_1<_x y$, we have $\M',\auf x,r_1\zu\models\Dv^+(\neg\psivar\land\psi^\dagger)$. 
as required.

\smallskip
$\Leftarrow$:
Suppose that
$\M,\auf r_0,r_1\zu\models\expform\land\phi^\dagger$ in some model $\M=\auf \Hh_{\F,\overline{G}},\mu\zu$ based on an expanding 2-frame
$\Hh_{\F,\overline{\G}}$ where $\F=\auf W,R\zu$ is transitive and for every $x\in W$, $\G_x=\auf W_x,<_x\zu$ is
a linear order.
Then $W_x\subseteq W_y$ and $<_x\subseteq <_y$ whenever $xRy$, $x,y\in W$. 
We may assume that $r_0$ is a root in $\F$, and so $r_1\in W_x$ for all $x\in W$.
Moreover, we may also assume that $r_1$ is a root in $\auf W_x,<_x\zu$ for every $x\in W$.
Let $\Sigma_x'=\auf W_x,\ne\zu$, for $x\in W$.
Then clearly $\Hh_{\F,\overline{G}'}\in \CC\eprod\cdiff$. We define a model $\M'=\auf \Hh_{\F,\overline{G}'},\mu'\zu$ 
by taking $\mu'(\pred):=\mu(\pred)$ for all $\pred\in\subf$.

We show by induction on $\psi$ that for all $\psi\in\subf$, $x\in W$, $u\in W_x$,
\begin{equation}\label{expihtwo}
\M,\auf x,u\zu\models\psi^\dagger\qquad\mbox{iff}\qquad\M',\auf x,u\zu\models\psi.
\end{equation}
Again, the only interesting case is that of $\Dv$. Suppose first that $\M,\auf x,u\zu\models(\Dv\psi)^\dagger$.
If 
\begin{equation}\label{psivar}
\M,\auf x,u\zu\models\psivar,
\end{equation}
then $r_1<_x u$ by the first conjunct of $\expform$, and so
$\M,\auf x,r_1\zu\models\Dv\psivar$. So
$\M,\auf x,r_1\zu\models\Dv^+(\neg\psivar\land\psi^\dagger)$
follows by the third conjunct of $\expform$. So there is $v\in W_x$ with $\M,\auf x,v\zu\models\neg\psivar\land\psi^\dagger$,
and so $v\ne u$ by \eqref{psivar}. Also, by the IH, we have $\M',\auf x,v\zu\models\psi$, and so
$\M',\auf x,v\zu\models\Dv\psi$ follows as required.
The other case when $\M,\auf x,u\zu\models\Dv\psi^\dagger$ is straightforward.

Conversely, suppose that $\M',\auf x,u\zu\models\Dv\psi$. Then there is $v\in W_x$, $v\ne u$ with $\M',\auf x,v\zu\models\psi$,
and so by the IH, $\M,\auf x,v\zu\models\psi^\dagger$. If $u<_x v$ then $\M,\auf x,u\zu\models\Dv\psi^\dagger$ follows.
If $v<_x u$ then by the second conjunct of $\expform$, we have $\M,\auf x,v\zu\models\Bv\psivar$, and so
$\M,\auf x,u\zu\models\psivar$ follows.
\end{proof}

Next, we show that $\{\auf\omega,<\zu\}\eprod\cdiff$-satisfiability has the `finite expanding second components property':

\begin{claim}\label{c:fmpe}
For any formula $\phi$,
if $\phi$ is $\{\auf\omega,<\zu\}\eprod\cdiff$-satisfiable, then $\phi$ is 
\mbox{$\{\auf\omega,<\zu\}\eprod\cdifff$}-satisfiable.
\end{claim}

\begin{proof}
Suppose $\M,\auf 0,r\zu\models\phi$ for some model $\M$ based on an expanding 2-frame
$\Hh_{\auf\omega,<\zu,\overline{\G}}$ where $\G_n=\auf W_n,\ne\zu$ are difference frames, for $n<\omega$.
For all $n<\omega$, $X\subseteq W_n$, we define $\cl_n(X)$ as the smallest set $Y$ such that $X\subseteq Y\subseteq W_n$ and having the following property:
If $x\in Y$ and $\M,\auf n,x\zu\models\Dv\psi$ for some $\psi\in\subf$, then there is $y\in Y$ such 
that $y\ne x$ and $\M,\auf n,y\zu\models\psi$.
It is not hard to see that if $X$ is finite then 
$|\cl_n(X)|\leq |X|+ 2|\subf|$.
Now define $\G_n':=\auf W_n',\ne\zu$ by taking $W_0':=\cl_0(\{r\})$ and  $W_{n+1}':=\cl_{n+1}(W_n')$ for
$n<\omega$.
Let $\M'$ be the restriction of $\M$ to the expanding 2-frame $\Hh_{\auf\omega,<\zu,\overline{\G}'}$.
A straightforward induction shows that for all $\psi\in\subf$, $n<\omega$, $w\in W_n'$, we have
$\M,\auf n,w\zu\models\psi$ iff \mbox{$\M',\auf n,w\zu\models\psi$.}
\end{proof}

Now Theorems~\ref{t:omegaeu} and \ref{t:decefinu}, respectively, follow from Claims~\ref{c:difftolin}, \ref{c:fmpe} and the following results:
\begin{itemize}\itemsep=3pt
\item
\cite[Thm.1]{KonevWZ05}
$\{\auf\omega,<\zu\}\eprod\clinf$-satisfiability is co-r.e.
\item
\cite[Thm.1]{gkwz06} 
$\clinf\eprod\,\clin$-satisfiability is decidable.
\end{itemize}


\section{Open problems}\label{disc}

Our results identify a limit beyond which the one-variable fragment of first-order linear temporal logic
is no longer decidable.
We have shown that ---unlike in the case of the two-variable fragment of classical first-order logic---
the addition of even limited counting capabilities ruins decidability in most cases: The resulting logic
\logic\ is very complex over various classes of linear orders, whenever the models
have constant, decreasing, or expanding domains. By generalising our techniques to the propositional bimodal setting, we have shown that the bimodal logic $\lindiffcomm$
of commuting \preorder and pseudo-equivalence relations is undecidable.
Here are some related unanswered questions:
\begin{enumerate}
\item
Is the bimodal logic $[\Kfour,\Diff]$ of commuting transitive and pseudo-equivalence relations decidable?
Is the product logic $\Kfour\mprod\Diff$ decidable?
As $\Kfour$ can be seen as a notational variant of the fragment of branching time logic $CTL$ that allows only two temporal operators  $E\D_F$ and its dual $A\B_F$, there is another reformulation of the second question: Is the one-variable fragment of first-order $CTL$ decidable when extended with counting
and when only $E\D_F$ and $A\B_F$ are allowed as temporal operators?
Note that without counting this coincides with $\Kfour\mprod\Sfive=[\Kfour,\Sfive]$-satisfiability,
and that is shown to be decidable by Gabbay and Shehtman \cite{Gabbay&Shehtman98}.

\item
Is \logic-satisfiability recursively enumerable in expanding domain models over the class of all linear orders?
The bimodal reformulation of this question: Is $\clin\eprod\cdiff$-satisfiability recursively enumerable?
By Cor.~\ref{co:linefou}, a positive answer would imply decidability of these.
Is \logic-satisfiability decidable in expanding domain models over $\auf\mathbb Q,<\zu$ or
$\auf\mathbb R,<\zu$?

\item
In decreasing 2-frames only `half' of commutativity ($\Bv\Bh \pred\to \Bh\Bv \pred$) is valid.
While in Theorem~\ref{t:main} we generalised Theorem~\ref{t:mainirr} to classes of decreasing 
\mbox{2-frames} and showed that
$\clin\dprod\cdiff$-satisfiability is undecidable, it is not clear whether the same can be done in the `abstract' setting: Is satisfiability undecidable in the class of 2-frames having
half-commuting \preorder and pseudo-equivalence relations?
\end{enumerate}
In our lower bound proofs we used reductions of counter machine problems.
Other lower bound results about bimodal logics with grid-like models use reductions of tiling or Turing 
machine problems \cite{Reynolds&Z01,gkwz03,gkwz05a}.
On the one hand, it is not hard to re-prove the same results using counter machine reductions. 
On the other, it seems tiling and Turing machine techniques require more control over the 
$\omega\times\omega$-grid than the limited expressivity that \logic\ provides. 
In order to understand the boundary of each technique, it would be interesting to find tiling or Turing 
machine reductions for the results of this paper.



\bibliographystyle{plain}

\end{document}